\documentclass[12pt]{article}
\usepackage[sectionbib,round]{natbib}
\usepackage{array,epsfig,fancyhdr,rotating}
\usepackage{sectsty,secdot}
\sectionfont{\fontsize{12}{14pt plus.8pt minus .6pt}\selectfont}
\subsectionfont{\fontsize{12}{14pt plus.8pt minus .6pt}\selectfont}
\usepackage{amsmath,amssymb,amsfonts,amsthm,multirow}
\usepackage[utf8]{inputenc}
\usepackage{graphicx,longtable,booktabs,verbatim}
\usepackage{algorithm,algpseudocode,bm,indentfirst}
\usepackage{needspace}

\newenvironment{breakablealgorithm}[1]{
    \par\addvspace{\intextsep}
    \Needspace{5\baselineskip}
    \refstepcounter{algorithm}
    \addcontentsline{loa}{algorithm}{\protect\numberline{\thealgorithm}#1}
    \hrule height .8pt\relax
    \nobreak\vskip 2pt
    \noindent\textbf{Algorithm \thealgorithm}\quad #1\par
    \nobreak\vskip 2pt
    \hrule height .4pt\relax
    \nobreak\vskip 2pt
}{
    \par\nobreak\vskip 2pt
    \hrule height .8pt\relax
    \addvspace{\intextsep}
}
\usepackage[hidelinks]{hyperref}
\newcommand{\E}{\mathbb E}
\newcounter{condition}
\newcommand{\condition}[2]{
    \refstepcounter{condition}
    \textbf{Condition \thecondition.} #2 \label{#1}
}
\newtheorem{theorem}{Theorem}
\newtheorem{lemma}{Lemma}
\newtheorem{corollary}[theorem]{Corollary}
\newtheorem{definition}{Definition}[section]
\newtheorem{quotedlemma}{Lemma}
\newtheorem{quotedtheorem}{Theorem}
\renewcommand{\sectionmark}[1]{}
\renewcommand{\subsectionmark}[1]{}
\hypersetup{pdftitle={Factor Model Estimation for High-Dimensional Time Series with Heteroskedastic Noise},pdfauthor={Da Zhao and Jing Chen}}
\graphicspath{{figures/}}
\numberwithin{equation}{section}
\begin{document}
\renewcommand{\baselinestretch}{1}
\fontsize{12}{14pt plus.8pt minus .6pt}\selectfont
\begingroup
\renewcommand{\baselinestretch}{1}
\fontsize{12}{14pt plus.8pt minus .6pt}\selectfont
\begin{center}
{\large\bfseries Factor Model Estimation for High-Dimensional\\[2pt]
Time Series with Heteroskedastic Noise\par}
\vspace{.4cm}
Da Zhao$^{1}$ and Jing Chen$^{2}$\\[.4cm]
{\itshape $^{1}$Yale School of Public Health, Yale University\\[2pt]
$^{2}$National University of Singapore}
\end{center}
\endgroup
\vspace{.15cm}
{\fontsize{9}{11.5pt plus.8pt minus.6pt}\selectfont
\begin{quotation}
\noindent\textit{Abstract:} Factor modeling provides a framework for extracting common components from high-dimensional time series. We develop an estimator of the factor loading space that combines contemporaneous covariance with lagged autocovariance information. Under cross-sectionally heteroskedastic white noise, the contemporaneous covariance contains heterogeneous diagonal noise contributions. We apply the existing HeteroPCA algorithm to the combined estimation matrix to correct this contamination while retaining information from serial dependence. The framework also accommodates noise contamination on a known set of entries subject to sparsity and incoherence conditions. Under regularity conditions for dependent observations, we establish consistency and derive convergence rates, identifying regimes in which incorporating contemporaneous covariance improves the rate relative to the lagged-autocovariance estimator. Simulations evaluate loading-space estimation error and the effects of contemporaneous covariance inclusion and diagonal correction. An application to S\&P 500 stock returns illustrates the method.

\vspace{9pt}
\end{quotation}}
\fontsize{12}{14pt plus.8pt minus .6pt}\selectfont

\section{Introduction}

\indent The modern capacity for data collection has led to an abundance of time series data, with high-dimensional time series becoming increasingly prevalent across diverse fields such as economics, finance, engineering, environmental science, medical research, network traffic monitoring, image processing, and beyond. Factor modeling has emerged as a powerful tool for analyzing such data from a dimension-reduction perspective, with foundational contributions including \cite{bai2002determining}, \cite{bai2003inferential}, and \cite{stock2005implications}. Related work studies dynamic specifications and the temporal structure of common components (e.g., \cite{molenaar1985dynamic}, \cite{tiao1989model}, \cite{barigozzi2016generalized}). A critical component in constructing a factor model is determining the number of factors, a topic that has been extensively studied in the context of vector factor models (see \cite{bai2002determining}, \cite{lam2012factor}, \cite{hallin2007determining}, \cite{amengual2007consistent}). More recently, this line of research has extended to matrix factor models (\cite{yu2022projected}, \cite{chen2023statistical}, \cite{wang2019factor}, \cite{yu2024dynamic}) and tensor factor models (\cite{chen2024rank}, \cite{chen2024estimation}, \cite{chen2022factor}).

While numerous studies have explored various forms of dynamic factor models—including vector, matrix, and tensor formulations—the lagged autocovariance framework of \cite{lam2012factor} does not exploit the potentially informative contemporaneous covariance. In this work, we focus on the vector factor model and study how contemporaneous and lagged covariance information can be used jointly for loading-space estimation.

Considering the factor model proposed by \cite{lam2012factor}, we observe a linear dynamic structure in the process $\bm{y}_t$, which can be decomposed as:
\begin{equation}
\label{factor}
\bm{y}_t = \bm{A}\bm{x}_t + \bm{\epsilon}_t,
\end{equation}
where $\bm{x}_t$ is an unobserved $r \times 1$ latent factor process with $r \leq p$, $\bm{A}$ is an unknown $p \times r$ loading matrix of full column rank, and $\bm{\epsilon}_t \sim WN(\bm{\mu}_\epsilon, \bm{\Sigma}_\epsilon)$ is a white-noise vector process. This model achieves effective dimension reduction when $r \ll p$, and the process $\bm{x}_t$ is referred to as the factor process. The key objective in inference is to estimate the factor loading space $\mathcal{M}(\bm{A})$, defined as the $r$-dimensional linear subspace spanned by the columns of $\bm{A}$. Although the loading matrix depends on the choice of factor coordinates, its column space is invariant under invertible transformations of those coordinates.
Define $\bm{\Sigma}_y(k) = \operatorname{Cov}(\bm{y}_{t+k}, \bm{y}_t)$, $\bm{\Sigma}_x(k) = \operatorname{Cov}(\bm{x}_{t+k}, \bm{x}_t)$, and $\bm{\Sigma}_{x\epsilon}(k) = \operatorname{Cov}(\bm{x}_{t+k}, \bm{\epsilon}_t)$ for $k \geq 0$. Under the standard assumption that future white noise is uncorrelated with current and past factors, for $k\geq1$ we have $\bm{\Sigma}_y(k)=\bm{A}\{\bm{\Sigma}_x(k)\bm{A}^\top+\bm{\Sigma}_{x\epsilon}(k)\}$, and hence $\mathcal{M}\{\bm{\Sigma}_y(k)\}\subseteq\mathcal{M}(\bm{A})$ for each $k\geq1$. The autocovariance-based method aggregates the information across lags through
\begin{equation}
\label{M}
\bm{M} = \sum_{k=1}^{k_0} \bm{\Sigma}_y(k) \bm{\Sigma}_y(k)^\top.
\end{equation}
As shown in \cite{lam2012factor}, under the rank condition $\operatorname{rank}(\bm{M})=r$, the factor loading space $\mathcal{M}(\bm{A})$ is spanned by the eigenvectors of $\bm{M}$ corresponding to its $r$ nonzero eigenvalues. Therefore, an estimator of $\mathcal{M}(\bm{A})$ is given by the span of the $r$ leading orthonormal eigenvectors of the empirical matrix $\widehat{\bm{M}}=\sum_{k=1}^{k_0} \widehat{\bm{\Sigma}}_y(k) \widehat{\bm{\Sigma}}_y(k)^\top$, where $\widehat{\bm{\Sigma}}_y(k)=(n-k)^{-1}\sum_{t=1}^{n-k}(\bm{y}_{t+k}-\bar{\bm{y}})(\bm{y}_t-\bar{\bm{y}})^\top$ and $\bar{\bm{y}}=n^{-1}\sum_{t=1}^n\bm{y}_t$.
However, this construction of $\bm{M}$ only incorporates lagged autocovariances with $k \geq 1$, thereby excluding the potentially informative contemporaneous covariance $\bm{\Sigma}_y(0)$. Even when the aggregate rank condition holds, the lagged signal may be weak in some factor directions. In such settings, excluding a comparatively informative $\bm{\Sigma}_y(0)$ may lead to less accurate recovery of the loading space.

In this paper, we propose an estimator of the factor loading space $\mathcal{M}(\bm{A})$ that incorporates information from both the contemporaneous covariance $\bm{\Sigma}_y(0)$ and the lagged covariances used in standard autocovariance-based methods. Integrating $\bm{\Sigma}_y(0)$ can improve loading-space estimation when it contains additional factor signal that is otherwise ignored. The primary challenge in this extension is heteroskedastic noise: the entries of the noise vector $\bm{\epsilon}_t$ may have substantially different variances, which contaminate the diagonal of the contemporaneous covariance.

Heteroskedasticity has long been studied in approximate factor models, including through likelihood-based estimation \citep{bai2016efficient}. For principal subspace estimation, \cite{zhang2022heteroskedastic} introduced HeteroPCA, which iteratively imputes the diagonal of a low-rank signal matrix using its off-diagonal structure. Related developments include inference with missing data \citep{yan2024inference}, Deflated HeteroPCA \citep{zhou2025deflated}, and diagonal-deletion approaches \citep{lounici2013sparse,cai2021subspace}.

A particularly relevant application is \cite{kaebi2026covariance}, who applies HeteroPCA to the contemporaneous sample covariance of financial returns and evaluates the resulting latent factors for asset pricing. That covariance-based construction uses $\widehat{\bm{\Sigma}}_y(0)$ as the input to HeteroPCA. Our estimator instead applies HeteroPCA to $\widehat{\bm{M}}^*=\widehat{\bm{\Sigma}}_y(0)+\widehat{\bm{M}}$, incorporating serial dependence directly into the matrix used for loading-space estimation. This distinction concerns the information used by the estimator; using time-indexed observations or rolling estimation windows does not by itself introduce lagged autocovariances. HeteroPCA applied to contemporaneous covariance alone is therefore a distinct and relevant benchmark. The comparison with the lagged-only estimator studied here addresses the benefit of adding contemporaneous information and does not, by itself, establish superiority over that contemporaneous-only benchmark.

Motivated by this distinction, we incorporate HeteroPCA \citep{zhang2022heteroskedastic} into the lagged-autocovariance framework of \cite{lam2012factor}. We are not aware of existing work that combines HeteroPCA with this framework to estimate the loading space of high-dimensional vector time series from contemporaneous and lagged covariances jointly. The methodological contribution lies in the joint time-series estimation matrix, its theoretical analysis under dependent observations, and the tradeoff between contemporaneous and lagged signal strengths; the HeteroPCA imputation algorithms themselves are taken from \cite{zhang2022heteroskedastic}. The main contributions are fourfold. First, we construct a joint estimation matrix that augments the Lam--Yao lagged-autocovariance matrix with contemporaneous covariance. Second, we apply HeteroPCA to correct the heterogeneous diagonal noise contributions in this matrix. Third, we apply the existing Generalized HeteroPCA algorithm of \cite{zhang2022heteroskedastic} to the joint time-series estimation matrix when the corrupted entries lie in a known set, under sparsity and incoherence conditions. This extension concerns the time-series estimation problem rather than a new imputation algorithm. Fourth, under assumptions on temporal dependence and signal strength, we establish consistency and derive convergence rates, identifying regimes in which adding contemporaneous covariance improves the rate relative to the lagged-autocovariance estimator.

The remainder of the paper is organized as follows. Section~\ref{sec:estimation} introduces the model and outlines our estimation procedure, which incorporates both contemporaneous and lagged covariance structures via the HeteroPCA algorithm. It also describes the application of the existing Generalized HeteroPCA algorithm to noise contamination on a known set of entries. Section~\ref{sec:consistency} establishes the theoretical properties of our estimator and compares its convergence rate with that of the lagged-autocovariance estimator. Section~\ref{sec:simulation} reports two simulation studies comparing the empirical performance of our method with existing approaches. In Section~\ref{sec:realdata}, we apply our method to S\&P 500 stock returns and compare the estimated loading-space structures. Finally, Section~\ref{sec:Discussion} concludes the paper.

\section{Methodology}
\label{sec:estimation}
\subsection{Notation and Preliminaries}
We use lowercase non-bold letters (e.g., \( x, y \)) to denote scalars, lowercase bold letters (e.g., \( \bm{x}, \bm{y} \)) to denote vectors, and uppercase bold letters (e.g., \( \bm{X}, \bm{Y} \)) to denote matrices. For any matrix $\bm H$, let $\sigma_j(\bm H)$ denote its $j$th largest singular value, including zero singular values. We reserve $\lambda_j(\bm H)$ for the $j$th largest eigenvalue of a symmetric matrix. The spectral norm and Frobenius norm are defined by $\|\bm H\|=\sup_{\|u\|_2=1}\|\bm H u\|_2=\sigma_1(\bm H)$ and $\|\bm H\|_F=(\sum_{i,j}H_{ij}^2)^{1/2}=(\sum_j\sigma_j(\bm H)^2)^{1/2}$. We write $\|\bm H\|_r=\sigma_r(\bm H)$, while $\|\bm H\|_{\min}$ denotes the smallest nonzero singular value of a nonzero matrix $\bm H$. For symmetric positive semidefinite matrices, $\sigma_j(\bm H)=\lambda_j(\bm H)$. For vectors, $\|\cdot\|$ denotes the Euclidean norm. For a symmetric square matrix \( \bm{X} \), let \( \Delta(\bm{X}) \) denote the matrix formed by zeroing out the diagonal entries of \( \bm{X} \), i.e., the \textit{off-diagonal} matrix of \( \bm{X} \), and $D(\bm{A})$ be $\bm{A}$ with all off-diagonal entries set to zero. We define the Stiefel manifold \( \bm{O}_{p,r} = \{ \bm{U} \in \mathbb{R}^{p \times r} : \bm{U}^\top \bm{U} = \bm{I}_r \} \), which consists of all \( p \times r \) matrices with orthonormal columns. For \( \bm{U} \in \bm{O}_{p,r} \), let \( \bm{U}^c \in \bm{O}_{p,p-r} \) denote its orthogonal complement, such that \( [\bm{U}, \bm{U}^c] \in \mathbb{R}^{p \times p} \) forms a complete orthogonal matrix.

Let \( e_i \) be the \( i \)-th canonical basis vector in \( \mathbb{R}^p \) (i.e., a vector with 1 in the \( i \)-th coordinate and 0 elsewhere). The \textit{incoherence constant} of \( \bm{U} \in \bm{O}_{p,r} \) is defined as
\begin{equation*}
    I(\bm{U}) = \frac{p}{r} \max_{i \in [p]} \| e_i^\top \bm{U} \|_2^2.
\end{equation*}
The incoherence constant quantifies the alignment between the singular vectors of a low-rank matrix and the standard basis. It captures how uniformly the signal energy is distributed across the coordinates. A low incoherence indicates that the energy of each singular vector is evenly spread out, while a high incoherence implies that the energy is concentrated in a few coordinates—an undesirable condition for accurate recovery guarantees.

Denote by $\mathcal{G} \subset [p] \times [p]$ a set of index pairs, where $[p]$ represents the set $\{1, 2, \dots, p\}$. For any matrix $\bm{H} \in \mathbb{R}^{p \times p}$, define $G(\bm{H})$ to be the matrix $\bm{H}$ with all entries but those in $\mathcal{G}$ set to zero. Let $\Gamma(\bm{H}) = \bm{H} - G(\bm{H})$ denote the complement projection.

We define
\begin{equation}
\label{eq:mask-amplification}
\eta = \sup_{\substack{\bm{H} \in \mathbb{R}^{p \times p},\,\bm H\ne0\\ \operatorname{rank}(\bm{H}) \leq 2r}} \frac{\|G(\bm{H})\|}{\|\bm{H}\|},
\end{equation}
which quantifies the maximum relative perturbation introduced by the entries in $\mathcal{G}$ on the singular subspace.

We further assume that the set of corrupted entries $\mathcal{G}$ is $b$-sparse in the sense that
\begin{equation*}
\max_i \left|\{j : (i, j) \in \mathcal{G}\}\right| \vee \max_j \left|\{i : (i, j) \in \mathcal{G}\}\right| \leq b,
\end{equation*}
i.e., each row and each column contains at most $b$ corrupted entries.

\subsection{Models and Estimation Approach}
Following the factor model in equation (\ref{factor}) proposed by \cite{lam2012factor}, we impose additional assumptions to enable estimation using the contemporaneous covariance. 

First, we allow the noise term $\bm{\epsilon}_t$ to be heteroskedastic. Under Condition \ref{C1} below, its components are mutually independent, so that $\bm{\Sigma}_\epsilon$ is diagonal, and their variances may differ across coordinates. No minimum degree of variance heterogeneity is required, and the homoskedastic case is also included.
We also assume that no nonzero linear combination of the components of $\bm{x}_t$ is white noise. This is an additional model restriction: absorbing a white-noise factor into $\bm{\epsilon}_t$ need not preserve the diagonal noise covariance required here. In addition, we assume that the loading matrix $\bm{A}$ has full column rank $r$; otherwise, model in equation (\ref{factor}) can be equivalently represented with a lower-dimensional factor structure.
Finally, because model in equation (\ref{factor}) is invariant under the transformation $(\bm{A}, \bm{x}_t) \mapsto (\bm{A}\bm{H}, \bm{H}^{-1} \bm{x}_t)$ for any invertible $r \times r$ matrix $\bm{H}$, we may, without loss of generality, assume that the columns of $\bm{A} = (\bm{a}_1, \cdots, \bm{a}_r)$ are orthonormal. That is, $\bm{A}^\top \bm{A} = \bm{I}_r$, where $\bm{I}_r$ denotes the $r \times r$ identity matrix.
We summarize these assumptions in the following Condition \ref{C1}:

\noindent \condition{C1}{In model in equation~(\ref{factor}), $\bm{\epsilon}_t \sim WN(\bm{\mu}_\epsilon, \bm{\Sigma}_\epsilon)$, where the elements of $\bm{\epsilon}_t$ are mutually independent, so that $\bm{\Sigma}_\epsilon$ is a diagonal matrix. 
No nonzero linear combination of the components of $\bm{x}_t$ is white noise. Furthermore, $\bm{A}^\top \bm{A} = \bm{I}_r$.}

We further impose a stationarity assumption, requiring that the future white noise components are uncorrelated with the current and past factors. This assumption is stated in Condition \ref{C2} below.

\noindent \condition{C2}{The factor process $\bm{x}_t$ is weakly stationary, and for any $k \geq 0$, $\mathrm{Cov}(\bm{x}_t, \bm{\epsilon}_{t+k}) = 0$.}

Under equation (\ref{factor}), Condition \ref{C1} and Condition \ref{C2}, we could get the contemporaneous covariance
\begin{equation*}
    \bm{\Sigma_y}(0)=\bm{A}\bm{\Sigma}_x(0)\bm{A}^\top+\bm{\Sigma}_\epsilon
\end{equation*}

\noindent Combining the contemporaneous covariance with the lagged covariance as $\bm{M}$ in equation (\ref{M}), we further define the matrix $\bm{M}^*$ to be

\begin{equation*}
    \bm{M}^*= \bm{\Sigma}_y(0)+\sum_{k=1}^{k_0} \bm{\Sigma}_y(k) \bm{\Sigma}_y(k)^{\top}
\end{equation*}
Here $\bm{M}^*$ is a $p\times p$ positive semidefinite matrix.

According to Condition \ref{C1}, in the presence of heteroskedastic noise, the differences in the bias terms on the diagonal elements of $\bm{\Sigma}_\epsilon$ lead to significant discrepancies and biases in the estimation. What is required is the estimation of $\bm{M}^*_0$ as follows:

\begin{equation*}
    \bm{M}^*_0= \bm{M}^*-\bm{\Sigma}_{\epsilon}=\bm{A}\bm{\Sigma}_x(0)\bm{A}^\top+\sum_{k=1}^{k_0} \bm{\Sigma}_y(k) \bm{\Sigma}_y(k)^{\top}
\end{equation*}

\noindent Based on Conditions \ref{C1} and \ref{C2}, it follows that:

\textit{The factor loading space $\mathcal{M}(\bm{A})$ is spanned by the eigenvectors of $\bm{M}_0^*$ corresponding to its nonzero eigenvalues, and the number of the nonzero eigenvalues is $r$.}

To estimate the factor loading space $\mathcal{M}(\bm{A})$, we construct the empirical lagged-autocovariance matrix and the joint estimation matrix as
\begin{equation*}
\widehat{\bm{M}}=\sum_{k=1}^{k_0}\widehat{\bm{\Sigma}}_y(k)\widehat{\bm{\Sigma}}_y(k)^\top,\qquad \widehat{\bm{M}}^*=\widehat{\bm{\Sigma}}_y(0)+\widehat{\bm{M}},
\end{equation*}
where $\widehat{\bm{\Sigma}}_y(k) = \frac{1}{n-k}\sum_{t=1}^{n-k}(\bm{y}_{t+k} - \bar{\bm{y}})(\bm{y}_{t} - \bar{\bm{y}})^\top$ denotes the sample covariance matrix of $\bm{y}_t$ at lag $k$, for $k=0,\ldots,k_0<n$, and $\bar{\bm{y}} = \frac{1}{n}\sum_{t=1}^n \bm{y}_t$ is the sample mean.

The estimation of the eigenvectors from $\widehat{\bm{M}}^*$ can often be inconsistent due to the presence of heteroskedastic noise in the error term. Specifically, when the diagonal elements of the noise covariance matrix $\bm{\Sigma}_\epsilon$ are equal, the estimated eigenspace aligns well with the true factor loading space $\mathcal{M}(\bm{A})$. However, under heteroskedasticity—where the variances of the noise components differ—the bias introduced in the diagonal entries of $\widehat{\bm{M}}^*$ can lead to significant distortions in the estimated eigenstructure, resulting in a substantial deviation from the true factor loading space.

To mitigate this bias, we adopt the HeteroPCA algorithm proposed by \cite{zhang2022heteroskedastic} and apply it to the matrix $\widehat{\bm{M}}^*$. The core idea of HeteroPCA is to iteratively refine the diagonal entries of the estimation matrix using information from the off-diagonal elements, which are less affected by heteroskedastic noise. By doing so, the algorithm effectively suppresses the impact of variance heterogeneity on the estimation process, leading to a more accurate recovery of the low-rank structure and a more consistent estimation of the factor loading space.

There exist numerous methods for estimating the number of factors $r$, including those proposed in \cite{bai2002determining}, \cite{hallin2007determining}, \cite{amengual2007consistent}, and \cite{lam2012factor}. In this work, our primary focus is on the accurate estimation of the factor loading space $\mathcal{M}(\bm{A})$. For consistency and comparability in evaluating improvements in space estimation, we adopt the estimator of $\hat{r}$ proposed by \cite{lam2012factor}, which is particularly suitable in the context of eigenanalysis-based procedures. The estimator is defined as
\begin{equation}
\label{r}
\hat{r} = \arg \min_{1 \leq i \leq R} \frac{\hat{\lambda}_{i+1}}{\hat{\lambda}_i},
\end{equation}
where $\hat{\lambda}_1 \geq \hat{\lambda}_2 \geq \cdots \geq \hat{\lambda}_p$ are the eigenvalues of $\widehat{\bm{M}}$, and $R$ is a prespecified integer upper bound satisfying $r<R<p$.

After obtaining the estimated rank $\hat{r}$ from equation (\ref{r}), we apply the HeteroPCA algorithm to the matrix $\widehat{\bm{M}}^*$. The detailed procedure is outlined in Algorithm \ref{algorithm1}. In both algorithms below, the rank-$\hat r$ approximation follows the truncated-SVD convention of \cite{zhang2022heteroskedastic}. For a symmetric iterate, we implement it by retaining the eigenvectors associated with the $\hat r$ largest absolute eigenvalues and keeping the signed eigenvalues in the reconstruction. We denote these magnitude-ordered signed eigenvalues by $d_j^{(i)}$, so that $|d_j^{(i)}|=\sigma_j(\bm N^{(i)})$; this ordering is distinct from the algebraic ordering denoted by $\lambda_j$. 

\begin{breakablealgorithm}{HeteroPCA}
\label{algorithm1}
\begin{algorithmic}[1] 
    \State \textbf{Input:} matrix $\widehat{\bm{M}}^*$, rank $\hat{r}$, maximum number of iterations $T$. 
    \State Initialize by setting the diagonal of $\widehat{\bm{M}}^*$ to be zero: $\bm{N}^{(0)}=\Delta(\widehat{\bm{M}}^*), i=0$. 
    \State \textbf{repeat:}
    \State \quad Compute the symmetric eigendecomposition:
    \State \quad $\bm{N}^{(i)}=\bm{U}^{(i)}\bm{\Lambda}^{(i)}\bm{U}^{(i)\top}=\sum_{j=1}^p d_j^{(i)}u_j^{(i)}u_j^{(i)\top}$
    \State \quad Order $|d_1^{(i)}|\geq\cdots\geq|d_p^{(i)}|$ and retain the signed eigenvalues:
    \State \quad $\bm{\tilde{N}}^{(i)}=\sum_{j=1}^{\hat{r}}d_j^{(i)}u_j^{(i)}u_j^{(i)\top}$
    \State \quad Update \\\quad $\bm{N}^{(i+1)}=D(\tilde{\bm{N}}^{(i)})+\Delta(\bm{N}^{(i)})$\\\quad $i=i+1$
    \State \textbf{Until} convergence or maximum number of iterations reached.
    \State \textbf{Output}: $\widehat{\bm{U}}=\bm{U}^{(T)}=[u_1^\top,\cdots,u_{\hat{r}}^\top]$, $\widehat{\bm{N}}=\bm{N}^{(T)}$
\end{algorithmic}
\end{breakablealgorithm}

Our estimation procedure applies the HeteroPCA algorithm of \cite{zhang2022heteroskedastic} to the joint time-series matrix, as summarized in Algorithm~\ref{algorithm1}, and consists of two main steps. First, we compute the matrix $\widehat{\bm{M}}^*$ and estimate the rank $\hat{r}$ using the method proposed in \cite{lam2012factor}. Second, we iteratively update $\widehat{\bm{M}}^*$ to correct the bias introduced by heteroskedastic noise in $\bm{\Sigma}_\epsilon$. This correction is performed by replacing the diagonal entries of $\widehat{\bm{M}}^*$ with those from a low-rank approximation, following the HeteroPCA procedure. After reaching the maximum number of iterations, the estimated factor loading matrix $\widehat{\bm{A}}_H$ is obtained as the matrix whose columns are the eigenvectors $\widehat{\bm{U}}$, and the final estimator of $\bm{M}^*_0$ is denoted by $\widehat{\bm{M}}_{H} = \widehat{\bm{N}}$. The formal definition of $\widehat{\bm{A}}_H$ is given below.

\begin{definition}[$\widehat{\bm{A}}_H$]
\label{estimationH}
Given the output $\widehat{\bm{U}}$ from Algorithm~\ref{algorithm1}, with input matrix $\widehat{\bm{M}}^*$, rank $\hat{r}$ calculated based on \cite{lam2012factor}, and maximum number of iterations $T$, we define $\widehat{\bm{A}}_H$ as the matrix whose columns are the eigenvectors $\widehat{\bm{U}}$.
\end{definition}

\subsection{Extension to the Generalized Case}

In the previous section, our model and estimation procedure were developed under the framework of the spiked covariance model with heteroskedastic noise. Specifically, in Condition~\ref{C1}, we assumed that $\bm{\epsilon}_t \sim WN(\bm{\mu}_\epsilon, \bm{\Sigma}_\epsilon)$, where the components of $\bm{\epsilon}_t$ are mutually independent, implying that $\bm{\Sigma}_\epsilon$ is a diagonal matrix. 

Here, we extend the framework to a \emph{generalized spiked covariance model}, in which $\bm{\Sigma}_\epsilon$ need not be diagonal. Our theoretical results concern noise covariance matrices whose support is contained in a known symmetric, $b$-sparse index set $\mathcal G\subseteq[p]\times[p]$, subject to Condition~\ref{C6} and the other assumptions stated in the theorems. Thus, the theory accommodates structured non-diagonal noise covariance matrices, rather than arbitrary dense covariance matrices. Under this generalized setting, Condition~\ref{C1} is replaced by the following:

\noindent \condition{C3}{In the model defined by equation~(\ref{factor}), $\bm{\epsilon}_t \sim WN(\bm{\mu}_\epsilon, \bm{\Sigma}_\epsilon)$. 
No nonzero linear combination of the components of $\bm{x}_t$ is white noise. Furthermore, $\bm{A}^\top \bm{A} = \bm{I}_r$.}

Unlike the classical spiked covariance model, the generalized setting allows for a broader structure of the noise covariance matrix. We assume that the corrupted entries lie in a known index set $\mathcal{G} \subset [p] \times [p]$, which is not necessarily restricted to the diagonal. 

For this setting, we apply the existing Generalized HeteroPCA algorithm of \cite{zhang2022heteroskedastic} (their Algorithm~2) to the joint time-series estimation matrix $\widehat{\bm{M}}^*$. Algorithm~\ref{algorithm2} presents this application using our notation: the procedure iteratively imputes the entries indexed by $\mathcal{G}$. The imputation algorithm itself is not new; our focus is its use with contemporaneous and lagged covariance information and its statistical properties for dependent observations. In the special case where $\mathcal{G}$ corresponds to the diagonal entries, Algorithm~\ref{algorithm2} reduces to Algorithm~\ref{algorithm1}.

\begin{breakablealgorithm}{Generalized HeteroPCA}
\label{algorithm2}
\begin{algorithmic}[1] 
    \State \textbf{Input:} matrix $\widehat{\bm{M}}^*$, rank $\hat{r}$, maximum number of iterations $T$, corruption subset $\mathcal{G}\subset [p]\times [p]$. 
    \State Initialize by setting the entries indexed by $\mathcal G$ to zero: $\bm{N}^{(0)}=\Gamma(\widehat{\bm{M}}^*), i=0$. 
    \State \textbf{repeat:}
    \State \quad Compute the symmetric eigendecomposition:
    \State \quad $\bm{N}^{(i)}=\bm{U}^{(i)}\bm{\Lambda}^{(i)}\bm{U}^{(i)\top}=\sum_{j=1}^p d_j^{(i)}u_j^{(i)}u_j^{(i)\top}$
    \State \quad Order $|d_1^{(i)}|\geq\cdots\geq|d_p^{(i)}|$ and retain the signed eigenvalues:
    \State \quad $\bm{\tilde{N}}^{(i)}=\sum_{j=1}^{\hat{r}}d_j^{(i)}u_j^{(i)}u_j^{(i)\top}$
    \State \quad Update \\\quad $\bm{N}^{(i+1)}=G(\tilde{\bm{N}}^{(i)})+\Gamma(\bm{N}^{(i)})$\\\quad $i=i+1$
    \State \textbf{Until} convergence or maximum number of iterations reached.
    \State \textbf{Output}: $\widehat{\bm{U}}=\bm{U}^{(T)}=[u_1^\top,\cdots,u_{\hat{r}}^\top]$, $\widehat{\bm{N}}=\bm{N}^{(T)}$
\end{algorithmic}
\end{breakablealgorithm}

Similarly to the HeteroPCA algorithm (Algorithm~\ref{algorithm1}), the output of the Generalized HeteroPCA algorithm (Algorithm~\ref{algorithm2}) consists of two elements: the estimated eigenvectors \( \widehat{\bm{U}} \) and the matrix \( \widehat{\bm{N}} \). Upon reaching the maximum number of iterations, the estimated factor loading matrix \( \widehat{\bm{A}}_{\mathrm{GH}} \) is formed by taking the columns of \( \widehat{\bm{U}} \), and the final estimator of \( \bm{M}_0^* \) is denoted by \( \widehat{\bm{M}}_{\mathrm{GH}} = \widehat{\bm{N}} \). The formal definition of $\widehat{\bm{A}}_{GH}$ is given below.

\begin{definition}[$\widehat{\bm{A}}_{GH}$]
\label{estimationGH}
Given the output $\widehat{\bm{U}}$ from Algorithm~\ref{algorithm2}, with input matrix $\widehat{\bm{M}}^*$, rank $\hat{r}$ calculated based on \cite{lam2012factor}, maximum number of iterations $T$, and corruption subset $\mathcal{G}\subset [p]\times [p]$, we define $\widehat{\bm{A}}_{GH}$ as the matrix whose columns are the eigenvectors $\widehat{\bm{U}}$.
\end{definition}

\section{Theory}
\label{sec:consistency}
In this section, we investigate the asymptotic properties of the relevant estimators under a set of regularity conditions. Our focus is on the convergence rates of the estimator $\widehat{\bm{A}}_{GH}$ obtained from Algorithm \ref{algorithm2}. Algorithm \ref{algorithm1} can be viewed as a special case of Algorithm \ref{algorithm2} by setting $\mathcal{G}$ to correspond to the diagonal entries.

To begin the theoretical analysis of Algorithm \ref{algorithm2}, we first consider the case where the dimension $p$ is fixed and the sample size $n \to \infty$, followed by the scenario where both $n \to \infty$ and $p \to \infty$. We start by introducing the necessary regularity conditions. Let $\lambda_1, \ldots, \lambda_p$ denote the eigenvalues of the matrix $\bm{M}$.

\noindent \condition{C4}{$\bm{y}_t$ is strictly stationary and $\psi$-mixing with the mixing coefficients $\psi(\cdot)$ satisfying the condition that $\sum_{t\geq 1}t\psi(t)^{1/2}<\infty$. Furthermore, $E\{|\bm{y}_t|^4\}<\infty$ element-wisely.}

\noindent \condition{C5}{$\lambda_1>\cdots>\lambda_r>0=\lambda_{r+1}=\cdots=\lambda_p$.}

\noindent \condition{C6}{Assume $\mathcal{G}\subseteq [p]\times [p]$ is $b$-sparse, and let $\eta$ be defined by \eqref{eq:mask-amplification}. For a sufficiently small constant $c>0$, the population incoherence condition $\frac{I(\bm{A})\|\bm{M}^*_0\|}{\|\bm{M}_0^*\|_{r}}\leq \frac{cp}{\eta br(b\wedge r)}$ holds. In addition, the sample perturbation inequality $\eta\|\Gamma(\widehat{\bm{M}}^*)-\Gamma(\bm{M}^*_0)\|\leq c\|\bm{M}_0^*\|_{r}$ holds with probability tending to one as $n\to\infty$.}

Condition \ref{C4} is standard in the factor modeling literature (see \cite{lam2012factor}, \cite{chen2022factor}). Condition \ref{C5} assumes that the \( r \) nonzero eigenvalues are mutually distinct. Condition \ref{C6} imposes a population incoherence requirement that depends on the corruption pattern, the rank, and the ratio of the largest to the smallest nonzero signal eigenvalue. It also requires the projected sample perturbation to be small relative to the smallest nonzero signal eigenvalue. Although the exact value of \( \eta \) can be difficult to determine in general, it satisfies \( \eta \leq \sqrt{b \wedge (2r)} \) for all \( b \)-sparse sets \( \mathcal{G} \) (\cite{zhang2022heteroskedastic}).

Under the above regularity conditions, we can establish the convergence of the estimator \( \widehat{\bm{A}}_{GH} \) when \( p \) is fixed.

\begin{theorem}
\label{Hetep1}
Under Conditions \ref{C2}--\ref{C6}, consider a fixed population model with fixed $p$, $r$, and $k_0$, where $1\leq r<p$. Assume $\eta>0$, and that the known corruption set is symmetric and contains the support of the noise covariance:
\begin{equation*}
(i,j)\in\mathcal G\Longleftrightarrow(j,i)\in\mathcal G,
\qquad
\operatorname{supp}(\bm\Sigma_\epsilon)\subseteq\mathcal G.
\end{equation*}
First run Algorithm \ref{algorithm2} with the true rank $r$. Choose an integer iteration count $T\geq1$ satisfying
\begin{equation}
\label{eq:theorem1-iterations}
2^{-T}/\eta=O(n^{-1/2}).
\end{equation}
its output satisfies
\begin{equation*}
\min_{\bm R\in\bm O_{r,r}}
\|\widehat{\bm A}_{GH}\bm R-\bm A\|=O_p(n^{-1/2}).
\end{equation*}
If the rank is estimated, additionally assume $\Pr(\widehat r=r)\to1$.
On the correct-rank event the same aligned comparison applies, and the loading-space error, defined even when $\widehat r\ne r$, satisfies
\begin{equation*}
\|\widehat{\bm A}_{GH}\widehat{\bm A}_{GH}^{\top}
-\bm A\bm A^\top\|=O_p(n^{-1/2}).
\end{equation*}
\end{theorem}
\vspace{1\baselineskip}

Theorem \ref{Hetep1} shows that the convergence rate under fixed \( p \) matches that of \cite{lam2012factor}. The key distinction arises when \( p \) diverges with \( n \), which will be explored in the subsequent analysis. To proceed, we first introduce additional regularity conditions.

\noindent \condition{C7}{Let $0\leq\delta_0\leq\delta\leq1$. For $k=1,\ldots,k_0$, assume $\|\bm\Sigma_x(k)\|\asymp p^{1-\delta}\asymp\sigma_r\{\bm\Sigma_x(k)\}$, where $\sigma_r$ denotes the $r$th largest singular value, including zero singular values. At lag zero, assume $\|\bm\Sigma_x(0)\|\asymp p^{1-\delta_0}\asymp\lambda_{\min}\{\bm\Sigma_x(0)\}$.}

\noindent \condition{C8}{For $k=1,\ldots,k_0$, $\|\bm\Sigma_{x\epsilon}(k)\|=o(p^{1-\delta})$. At lag zero, $\bm\Sigma_{x\epsilon}(0)=0$ already follows from Condition \ref{C2}.}

\noindent \condition{C9}{The moment and mixing bounds in Condition \ref{C4} hold uniformly in the dimension: writing $\bm z_t=\bm y_t-E\bm y_t$ and $\psi_p(h)$ for the mixing coefficient at dimension $p$, assume $\sup_p\max_{1\leq a\leq p}E|z_{t,a}|^4\leq C$ and $\sup_p\sum_{h=1}^{\infty}h\,\psi_p(h)^{1/2}\leq C$ for a constant $C$ independent of $n,p$.}

Condition \ref{C7} separates contemporaneous factor strength, indexed by $\delta_0$, from lagged factor strength, indexed by $\delta$. These strength assumptions do not follow merely from $\bm A^\top\bm A=\bm I_r$. The restriction $\delta_0\leq\delta$ is consistent with stationarity, which implies $\|\bm\Sigma_x(k)\|\leq\|\bm\Sigma_x(0)\|$. Condition \ref{C8} controls the positive-lag factor--noise cross-covariances, not the contemporaneous signal strength. Condition \ref{C9} makes explicit the dimension-uniform assumptions used in Lemma \ref{Mdiff}.

\begin{theorem}
\label{Hetep2}
Under Conditions \ref{C2}--\ref{C9}, let $r$ and $k_0\geq1$ be fixed as $n,p\to\infty$. Assume the known set $\mathcal G$ is symmetric, $\operatorname{supp}(\bm\Sigma_\epsilon)\subseteq\mathcal G$, $\eta>0$, and Algorithm \ref{algorithm2} is run with the true rank $r$. Write $q_p=\max\{p^{2-2\delta},p^{1-\delta_0}\}$ and $\rho_{n,p}=p^{2-\delta}n^{-1/2}/q_p$. Suppose $p^\delta n^{-1/2}=O(1)$ and $\rho_{n,p}\to0$. Choose an integer iteration count $T\geq1$ such that $2^{-T}/\eta=O(\rho_{n,p})$. Then
\begin{equation}
\label{eq:theorem2-aligned-and-projector-rate}
\begin{aligned}
\min_{\bm R\in\bm O_{r,r}}\|\widehat{\bm A}_{GH}\bm R-\bm A\|
&=O_p\left(\min\{p^\delta n^{-1/2},
p^{\delta_0-\delta+1}n^{-1/2}\}\right),\\
\|\widehat{\bm A}_{GH}\widehat{\bm A}_{GH}^\top-\bm A\bm A^\top\|
&=O_p\left(\min\{p^\delta n^{-1/2},
p^{\delta_0-\delta+1}n^{-1/2}\}\right).
\end{aligned}
\end{equation}
\end{theorem}
The growth condition $p^\delta n^{-1/2}=O(1)$ absorbs the quadratic sampling term; the separate requirement $\rho_{n,p}\to0$ ensures a small perturbation relative to the signal eigengap. The stronger condition $p^\delta n^{-1/2}\to0$ suffices for both. The iteration condition controls the geometric remainder directly and also covers a zero projected perturbation. The result concerns subspace estimation, not an unaligned eigenvector matrix, and is an upper bound rather than a claim of an optimal rate.

For the lagged-autocovariance estimator $\widehat{\bm A}_M$ of \cite{lam2012factor}, the benchmark loading-space error bound is $\|\widehat{\bm A}_M\widehat{\bm A}_M^\top-\bm A\bm A^\top\|=O_p(p^\delta n^{-1/2})$ under the corresponding regularity and growth conditions. We use projection-matrix error here, consistently with the rotation-invariant conclusion of Theorem \ref{Hetep2}.

Under the assumptions of Theorem \ref{Hetep2}, the upper bound for the joint estimator is $O_p(\min\{p^\delta n^{-1/2},p^{1+\delta_0-\delta}n^{-1/2}\})$. The second term has a strictly smaller power of $p$ when $2\delta>1+\delta_0$, equivalently when $\delta>1/2$ and $\delta_0<2\delta-1$. This regime corresponds to weaker positive-lag factor covariance and a sufficiently strong contemporaneous factor covariance in the sense of Condition \ref{C7}. Otherwise, the bound has the same order as the lagged-only benchmark.

This comparison concerns the orders of the stated error upper bounds. It does not establish rate optimality, a matching lower bound for either estimator, or uniform finite-sample superiority of the proposed method.

Next, we consider Algorithm \ref{algorithm1} as the special case $\mathcal G=\{(i,i):1\leq i\leq p\}$. Then Algorithm \ref{algorithm2} reduces to Algorithm \ref{algorithm1}, and
\begin{equation}
\label{eq:diagonal-mask-parameters}
b=1,\qquad
\eta=\sup_{\substack{\bm H\in\mathbb R^{p\times p},\,\bm H\ne0\\
\operatorname{rank}(\bm H)\leq2r}}
\frac{\|D(\bm H)\|}{\|\bm H\|}=1.
\end{equation}
The upper bound follows from $|H_{ii}|\leq\|\bm H\|$, and equality is attained by $\bm H=e_i e_i^\top$. The population part of Condition \ref{C6} becomes the following condition; its sample perturbation part is verified in the proof of the corollary.

\noindent \condition{C10}{For a sufficiently small constant $c_I>0$, the population incoherence condition $\frac{I(\bm A)\|\bm M_0^*\|}{\|\bm M_0^*\|_r}\leq c_Ip/r$ holds.}

\begin{corollary}
\label{Hetep3}
Under Conditions \ref{C1}, \ref{C2}, \ref{C4}, \ref{C5}, and \ref{C7}--\ref{C10}, let $r$ and $k_0\geq1$ be fixed as $n,p\to\infty$. Run Algorithm \ref{algorithm1} with the true rank $r$. With $q_p$ and $\rho_{n,p}$ defined in Theorem \ref{Hetep2}, assume $p^\delta n^{-1/2}=O(1)$ and $\rho_{n,p}\to0$. Choose an integer $T\geq1$ such that $2^{-T}=O(\rho_{n,p})$. Then
\begin{equation}
\label{eq:corollary-diagonal-subspace-rate}
\begin{aligned}
\min_{\bm R\in\bm O_{r,r}}\|\widehat{\bm A}_H\bm R-\bm A\|
&=O_p\left(\min\{p^\delta n^{-1/2},
p^{1+\delta_0-\delta}n^{-1/2}\}\right),\\
\|\widehat{\bm A}_H\widehat{\bm A}_H^\top-\bm A\bm A^\top\|
&=O_p\left(\min\{p^\delta n^{-1/2},
p^{1+\delta_0-\delta}n^{-1/2}\}\right).
\end{aligned}
\end{equation}
\end{corollary}

Corollary \ref{Hetep3} gives the diagonal-noise specialization of Theorem \ref{Hetep2}, with the same rate under the corresponding assumptions.

\section{Simulation}
\label{sec:simulation}
We consider two simulation settings: (i) increasing sample size $n$ with fixed dimension $p$ and factor number $r$; and (ii) increasing $n$ and $p$ with fixed $r$. The comparisons are designed to assess the contribution of contemporaneous and lagged covariance information, as well as the effect of correcting heteroskedastic noise.

We introduce four loading-space estimators within a common comparison framework. Our proposed estimator $\widehat{\bm A}_H$ applies HeteroPCA to the joint matrix $\widehat{\bm M}^*$ and is defined in Definition~\ref{estimationH}. The three benchmarks are $\widehat{\bm A}_P$, obtained by uncorrected PCA of $\widehat{\bm M}^*$; $\widehat{\bm A}_M$, the lagged-only estimator of \cite{lam2012factor} based on $\widehat{\bm M}$; and $\widehat{\bm A}_0$, obtained by applying HeteroPCA to the contemporaneous sample covariance $\widehat{\bm\Sigma}_y(0)$, following the covariance-based construction in \cite{kaebi2026covariance}. Their definitions are given together below. The definitions use a common rank $\hat r$, which can be selected by the eigenvalue-ratio rule in~\eqref{r}, and a common iteration limit $T$ for the two HeteroPCA estimators. In the simulations below, we instead set $\hat r=r$ to isolate loading-space estimation from factor-number selection.

\begin{definition}[$\widehat{\bm A}_P$]
\label{estimationP}
We define $\widehat{\bm A}_P\in\mathbb R^{p\times\hat r}$ as the matrix whose columns are orthonormal eigenvectors of $\widehat{\bm M}^*$ corresponding to its $\hat r$ largest eigenvalues. No diagonal correction is applied.
\end{definition}

\begin{definition}[$\widehat{\bm A}_M$]
\label{estimationM}
We define $\widehat{\bm A}_M\in\mathbb R^{p\times\hat r}$ as the matrix whose columns are orthonormal eigenvectors of $\widehat{\bm M}$ corresponding to its $\hat r$ largest eigenvalues. Its input matrix uses only positive-lag sample autocovariances.
\end{definition}

\begin{definition}[$\widehat{\bm A}_0$]
\label{estimation0}
Apply Algorithm~\ref{algorithm1} with input matrix $\widehat{\bm\Sigma}_y(0)$ in place of $\widehat{\bm M}^*$, rank $\hat r$, and iteration limit $T$. We define $\widehat{\bm A}_0\in\mathbb R^{p\times\hat r}$ as the matrix whose columns are the resulting orthonormal loading-space estimates. Conditional on the chosen rank, its matrix input contains no positive-lag autocovariances.
\end{definition}

Comparing $\widehat{\bm A}_H$ with $\widehat{\bm A}_P$ isolates the effect of diagonal correction for the joint matrix. Comparing $\widehat{\bm A}_H$ with $\widehat{\bm A}_M$ assesses the effect of adding contemporaneous information with heteroskedasticity correction, while comparison with $\widehat{\bm A}_0$ assesses the additional contribution of positive-lag information when both estimators use HeteroPCA.

Since only the column space of $\bm A$ is identifiable and the loading matrices have orthonormal columns, we evaluate estimation accuracy through projection-matrix error. Following \cite{chen2022factor} and \cite{zhang2024factor}, we use the Frobenius norm and define
\begin{equation}
\label{eq:error_metric}
\begin{aligned}
    \text{Error}_{H} &= \| \bm A\bm A^\top - \widehat{\bm A}_{H}\widehat{\bm A}_{H}^\top \|_F,\\
    \text{Error}_{P} &= \| \bm A\bm A^\top - \widehat{\bm A}_{P}\widehat{\bm A}_{P}^\top \|_F,\\
    \text{Error}_{M} &= \| \bm A\bm A^\top - \widehat{\bm A}_{M}\widehat{\bm A}_{M}^\top \|_F,\\
    \text{Error}_{0} &= \| \bm A\bm A^\top - \widehat{\bm A}_{0}\widehat{\bm A}_{0}^\top \|_F.
\end{aligned}
\end{equation}
These errors are invariant to orthogonal rotations of the estimated loading matrices, and smaller values indicate more accurate loading-space estimation. Both experiments compare all four estimators using these errors.

\noindent\textbf{Experiment I.}
We set $p=10$, $r=1$, $k_0=1$, and $n\in\{20,40,\ldots,10000\}$, and generate the $10$-dimensional observations from
\begin{equation*}
\begin{aligned}
    \bm y_t&=\bm A\bm x_t+\bm\epsilon_t,\qquad
    \bm A\in\mathbb R^{10\times1},\quad \bm A^\top\bm A=1,\\
    \bm x_t&=0.98\bm x_{t-1}+\bm u_t,\qquad
    \bm u_t\overset{\mathrm{iid}}{\sim}N(0,0.0198p),\\
    \bm\epsilon_t\mid\bm\Sigma_\epsilon
    &\overset{\mathrm{iid}}{\sim}N_{10}(\bm0,\bm\Sigma_\epsilon).
\end{aligned}
\end{equation*}
We obtain $\bm A$ by QR-normalizing ten independent $N(0,1)$ draws and keep it fixed throughout. In each replication, $\bm x_1\sim N(0,0.0198p)$ is independent of later innovations; we generate $2n$ factor values and retain the last $n$. The stationary factor variance is $0.5p=5$, although initialization is not stationary. The diagonal $\bm\Sigma_\epsilon$ assigns variance 41 to three uniformly selected coordinates and variance 1 to the other seven. It is fixed within each replication and redrawn between replications; conditional on it, noise is independent across coordinates and time and independent of the factors and loadings. For each $n$, all four methods use the same 500 datasets and the true rank. Covariances use full-sample centering and denominator $n-k$. Both HeteroPCA methods use at most 50 iterations with tolerance $10^{-5}$ for relative Frobenius change; all runs are retained.

\begin{figure}[!htb]
    \centering
    \includegraphics[width=0.78\linewidth,height=0.32\textheight,keepaspectratio]{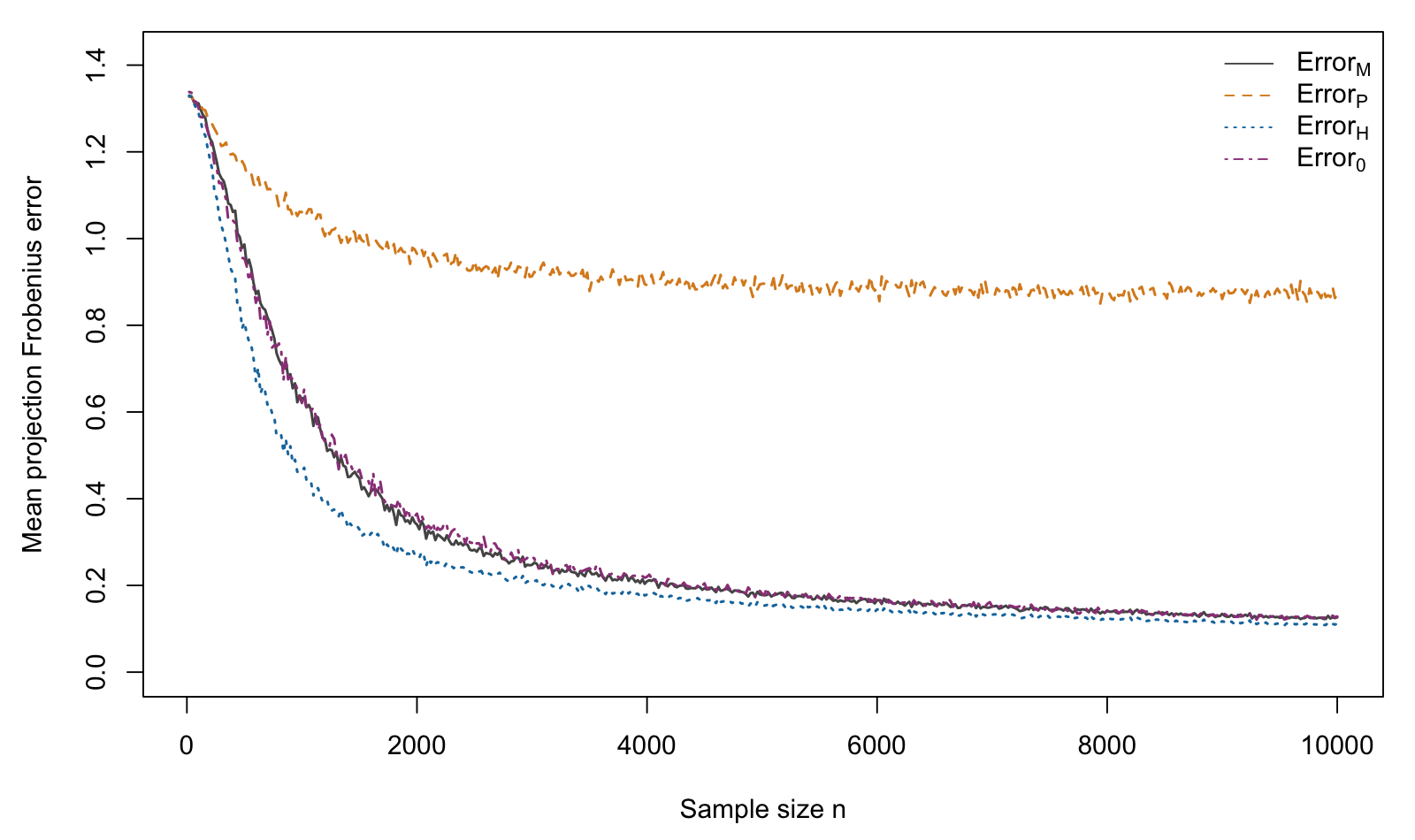}
    \caption{Experiment I: mean projection errors over 500 replications, with $p=10$, $r=1$, and $k_0=1$. Solid, dashed, dotted, and dot-dashed curves represent $\text{Error}_M$, $\text{Error}_P$, $\text{Error}_H$, and $\text{Error}_0$, respectively; see~\eqref{eq:error_metric}.}
    \label{fig:1}
\end{figure}

Figure~\ref{fig:1} shows a pronounced plateau in $\text{Error}_P$, which remains $0.8706$ at $n=10000$: uncorrected PCA retains the heterogeneous diagonal noise in $\bm M^*=\bm M_0^*+\bm\Sigma_\epsilon$, whose leading eigenspace differs from the target in this design. In contrast, $\text{Error}_M$, $\text{Error}_H$, and $\text{Error}_0$ decrease overall with $n$. The lagged-only estimator avoids this population contamination because the independent white noise contributes no positive-lag covariance. The two HeteroPCA estimators instead iteratively correct diagonal contamination, using the joint matrix for $\widehat{\bm A}_H$ and contemporaneous covariance alone for $\widehat{\bm A}_0$. Their declining errors are consistent with these mechanisms, but do not establish finite-sample unbiasedness. The joint HeteroPCA curve reaches low mean-error levels with fewer observations than either benchmark over most of the grid, consistent with combining the informative lagged signal from the persistent factor process with contemporaneous information and diagonal correction.

\textbf{Experiment II.}
We use model~\eqref{factor} with $r=3$, $n\in\{200,400\}$, $p/n\in\{0.2,0.5,0.8,1.2\}$, and $k_0\in\{1,2,3\}$, giving 24 configurations including $p>n$. In each replication, $\bm A$ is redrawn by orthonormalizing a standard Gaussian $p\times r$ matrix, and
\begin{equation*}
\begin{aligned}
    \bm x_t&=\operatorname{diag}(0.90,-0.94,0.98)\bm x_{t-1}+\bm u_t,\\
    \bm u_t&\overset{\mathrm{iid}}{\sim}
    N_3\!\left(\bm0,p\,\operatorname{diag}(0.057,0.03492,0.01188)\right).
\end{aligned}
\end{equation*}
The stationary factor covariance is $\bm\Sigma_x(0)=0.3p\bm I_3$. The Gaussian noise follows Experiment I, with exactly $30\%$ of coordinates randomly assigned variance 41 and the remainder variance 1, independently of the factors and loadings. We retain $n$ observations after a burn-in of $n$ and use 200 replications per configuration at the true rank. Within each $(n,p)$, the same datasets are used across methods and lag choices.

Figure~\ref{fig:exp2-four} shows the four mean projection errors across the $(k_0,n,p)$ configurations; detailed numerical results are reported in Table~\ref{tab1} in the supplementary appendix. All four mean errors decrease as $n$ increases from 200 to 400 at fixed $p/n$ and $k_0$, and as $p$ increases at fixed $n$ and $k_0$. For example, at $n=400$ and $k_0=3$, increasing $p$ from 80 to 480 reduces the mean $\text{Error}_H$ from $1.1545$ to $0.9269$. The proposed estimator has the smallest mean error in all 24 configurations.

Across the configurations in Table~\ref{tab1}, relative to $\widehat{\bm A}_M$, $\widehat{\bm A}_P$, and $\widehat{\bm A}_0$, the mean-error reductions for $\widehat{\bm A}_H$ range from $3.03\%$ to $13.36\%$, $3.05\%$ to $14.28\%$, and $0.65\%$ to $22.66\%$, respectively. These comparisons support the joint use of contemporaneous information, serial dependence, and diagonal correction in the tested configurations. The improvement with dimension occurs under the specified factor-strength scaling $\bm\Sigma_x(0)=0.3p\bm I_3$, rather than establishing that larger $p$ improves estimation in general.

\begin{figure}[!htb]
    \centering
    \includegraphics[width=\linewidth,height=0.50\textheight,keepaspectratio]{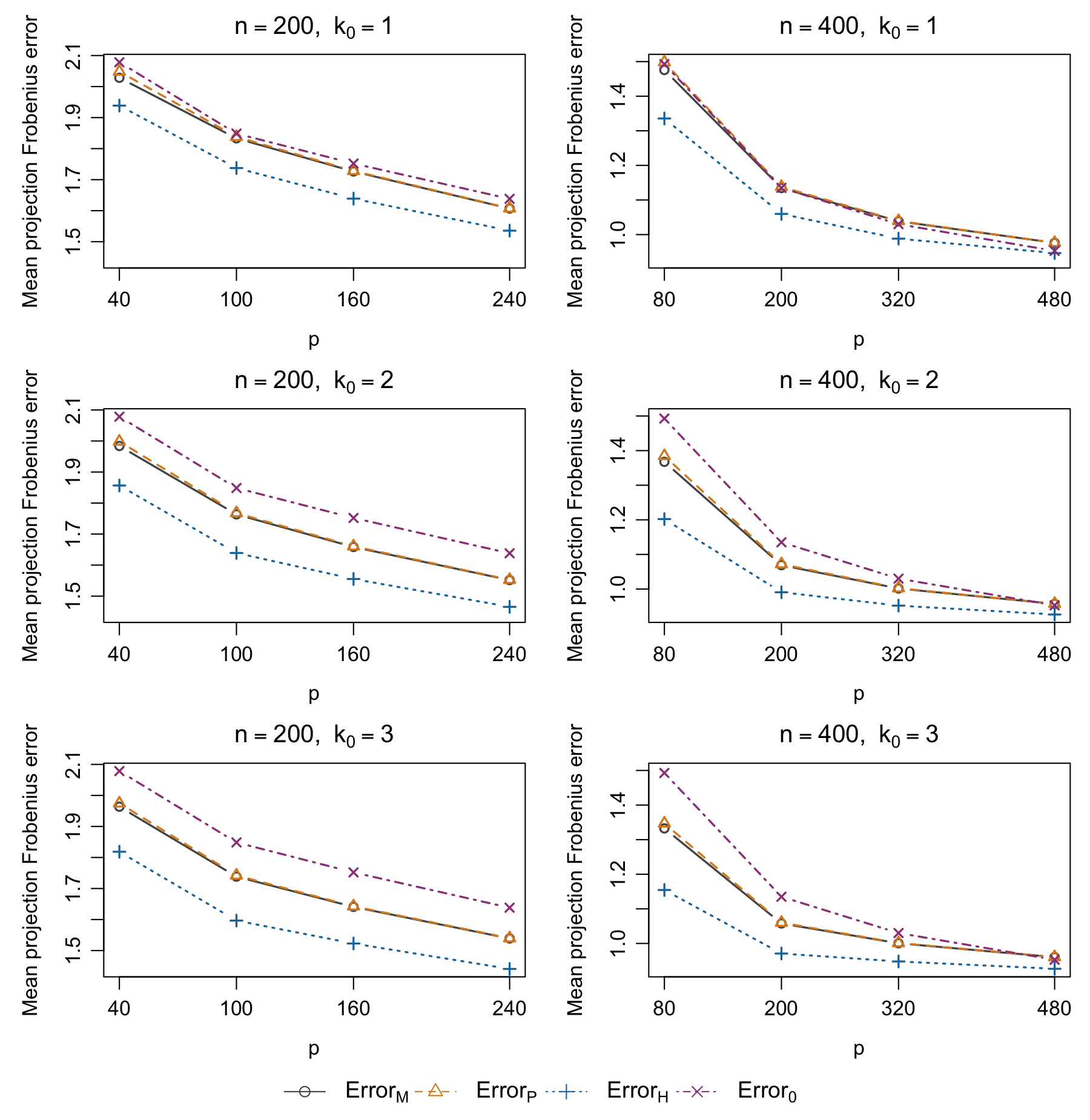}
    \caption{Experiment II: mean projection errors over 200 replications per configuration, with $r=3$. Columns correspond to $n=200$ and $n=400$; rows correspond to $k_0=1,2,3$. The four curves show $\text{Error}_M$, $\text{Error}_P$, $\text{Error}_H$, and $\text{Error}_0$ from~\eqref{eq:error_metric}. Each column uses a common vertical scale; lines connect the simulated settings without smoothing.}
    \label{fig:exp2-four}
\end{figure}

\section{Daily returns for S\&P 500 stocks}
\label{sec:realdata}
\label{sec:sp500}
Figures~\ref{fig:sp500-main-initial}--\ref{fig:sp500-main-sensitivity} display the three stock analyses. Full-size figures are provided in Section~\ref{sec:supp-realdata-figures} of the Supplementary Material.

The stock panel contains $p=110$ daily return series observed over $n=1634$ trading days, from January 3, 2002, to June 30, 2008. At each date, $\bm y_t$ collects the natural-log changes in the Yahoo Finance adjusted prices of these stocks. We use the archived returns in decimal units, without further differencing or coordinatewise variance standardization. Thus, differences in marginal return variability are retained. This is a fixed subset of stocks, not a reconstruction of the historical S\&P 500 membership at each date.

We examine the structure of the estimated loading spaces through varimax-rotated loading matrices. The stocks span 11 sectors, with the largest groups being Health Care (21 stocks), Information Technology (17), Financials (16), and Industrials (16). Sector labels are matched to the September 25, 2026, iShares IVV holdings file and used only to order the rows of the heatmaps. The archived holdings snapshot and the ticker-to-sector mapping accompany the analysis code. The source is the \href{https://www.ishares.com/us/products/239726/ishares-core-s-p-500-etf/latest-holdings.csv}{official iShares IVV holdings file}. The archived ticker MMC is matched to its renamed ticker MRSH for classification only; the return series is unchanged. These are current classifications rather than historical classifications for 2002--2008, so the sector interpretations below should be understood with this limitation. Sector labels do not enter any of the four estimators.

We first use the initial subsample of 1144 observations, ending July 19, 2006, with $k_0=1$. Each sample is centered separately, and its lag-$k$ covariance uses denominator equal to its length minus $k$. The eigenvalue-ratio rule of \citet{lam2012factor}, applied to $\widehat{\bm M}$ over candidate dimensions $1,\ldots,9$, selects $\hat r=2$, and all four estimators use this dimension. To examine the structure over the entire observation period, we also fit the full panel with $k_0=2$. Applying the same rank-selection rule to this full-sample matrix gives $\hat r=3$. This is a separate descriptive analysis using the entire observation period. Finally, we repeat the full-sample analysis with $k_0=1$ to assess sensitivity to the lag choice. All loading spaces are estimated directly from the return panel in a single stage. Both HeteroPCA methods use tolerance $10^{-5}$ and an iteration cap of 50.

Within each specification, we rotate the four estimated orthonormal loading bases using varimax and then align column signs and permutations to those of $\widehat{\bm A}_H$ for display. These operations leave the estimated loading spaces unchanged. Rows are ordered first by sector and then alphabetically by ticker, and all four method panels within each specification share a grayscale with symmetric limits. The columns are labelled F1, F2, and F3 only for reference: they are not uniquely identified economic factors or a ranking of factor strength. In each figure, the four blocks show $\widehat{\bm A}_M$, $\widehat{\bm A}_P$, $\widehat{\bm A}_H$, and $\widehat{\bm A}_0$ from left to right. Lighter and darker shades indicate lower and higher signed loadings; grayscale ranges differ across figures.

\begin{figure}[!htbp]
    \centering
    \includegraphics[page=1,width=0.80\linewidth,height=0.34\textheight,keepaspectratio]{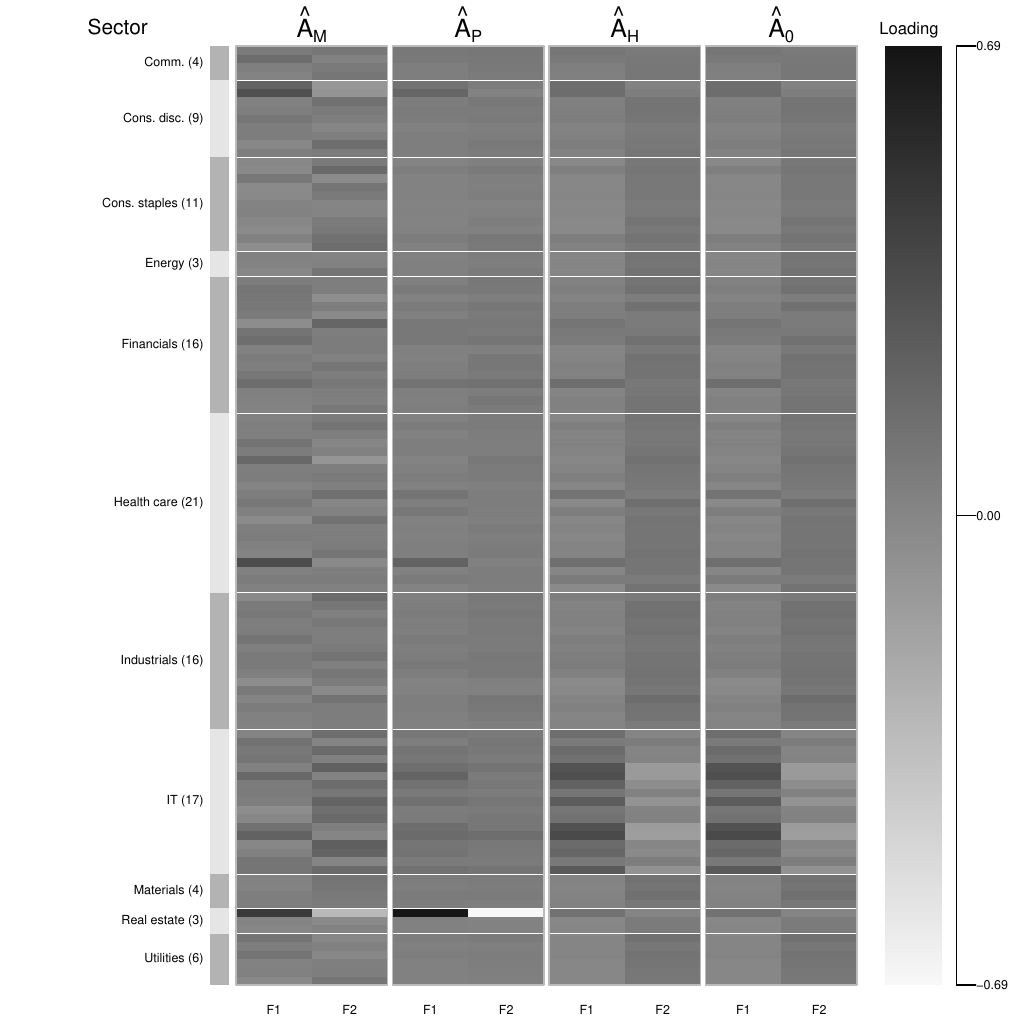}
    \caption{Stock loading matrices for the initial subsample ($n=1144$, $k_0=1$, $\hat r=2$).}
    \label{fig:sp500-main-initial}
\end{figure}

\begin{figure}[!htbp]
    \centering
    \includegraphics[page=2,width=0.80\linewidth,height=0.34\textheight,keepaspectratio]{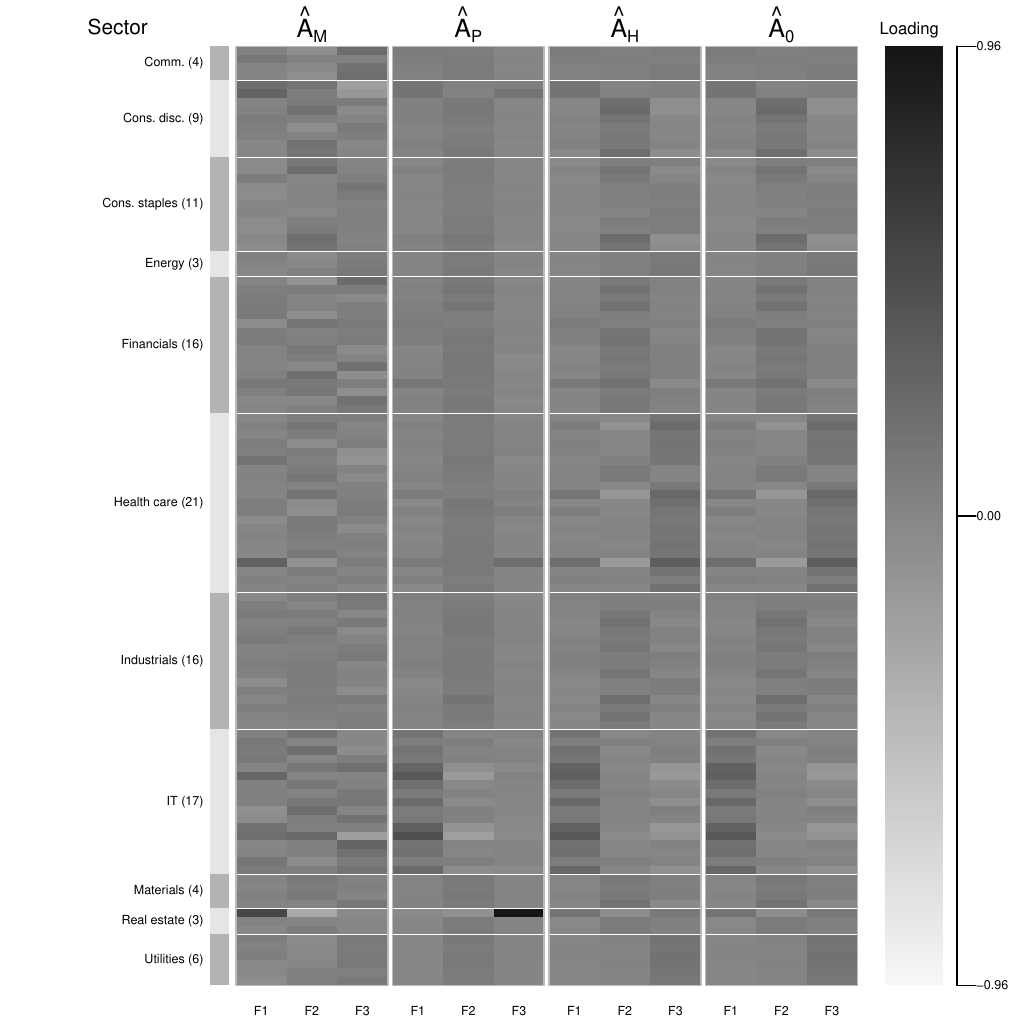}
    \caption{Stock loading matrices for the full sample with two lags ($n=1634$, $k_0=2$, $\hat r=3$).}
    \label{fig:sp500-main-full}
\end{figure}

\begin{figure}[!htbp]
    \centering
    \includegraphics[page=3,width=0.80\linewidth,height=0.34\textheight,keepaspectratio]{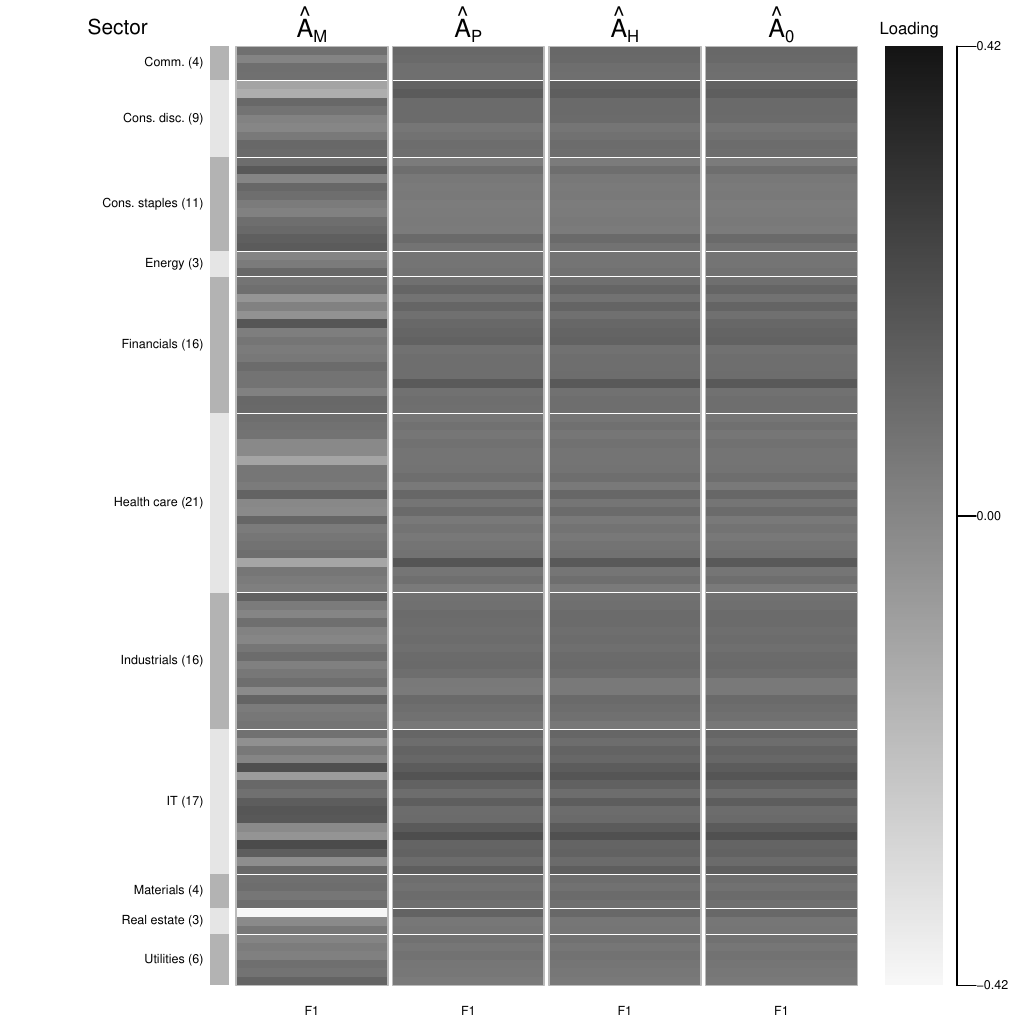}
    \caption{Stock loading matrices for the full sample with one lag ($n=1634$, $k_0=1$, $\hat r=1$).}
    \label{fig:sp500-main-sensitivity}
\end{figure}

Figure~\ref{fig:sp500-main-initial} displays the two-factor estimates from the initial training sample. For $\widehat{\bm A}_H$, the ten largest absolute loadings in F1 all correspond to Information Technology stocks, including NVDA, AMD, MU, AMAT, and TXN. F2 has sizeable loadings across several sectors, rather than representing a single sector. In contrast, the direct joint estimator $\widehat{\bm A}_P$ places substantial weight on one stock, EQIX. Its squared row norm, equivalently the corresponding diagonal entry of $\widehat{\bm A}_P\widehat{\bm A}_P^\top$, is $0.9224$, whereas the largest squared row norm for $\widehat{\bm A}_H$ is $0.1529$. Because this measure is invariant to orthogonal rotation, the difference reflects the estimated spaces rather than just the choice of displayed basis. The correction produces a less concentrated loading space in this sample, although the unobserved true loadings prevent us from interpreting concentration alone as estimation bias.

Figure~\ref{fig:sp500-main-full} shows the three-factor structure obtained from the full sample with $k_0=2$. For the proposed estimator, F1 again has large loadings on technology stocks, including NVDA, AMD, MU, AMAT, and TXN. F2 has relatively large loadings on LOW, TGT, HD, NSC, and TJX, spanning consumer and industrial businesses. F3 has large loadings on several health-care stocks, including REGN, GILD, AMGN, HUM, and UNH, together with stocks from other sectors. Thus, the estimated space admits a representation with sector-related patterns, but not a separation into three pure sector factors. The direct joint PCA estimate remains highly concentrated on EQIX, with a squared row norm of $0.9153$, compared with a maximum of $0.1806$ for the proposed estimator.

An important feature of the first two figures is the close agreement between $\widehat{\bm A}_H$ and $\widehat{\bm A}_0$. Their normalized projection distances, defined as $\|\widehat{\bm A}_H\widehat{\bm A}_H^\top-\widehat{\bm A}_0\widehat{\bm A}_0^\top\|_F/\sqrt{2\hat r}$, are $8.77\times10^{-5}$ in the initial training sample and $7.19\times10^{-4}$ in the full-sample two-lag analysis. Correspondingly, $\|\widehat{\bm M}\|/\|\widehat{\bm\Sigma}_y(0)\|$ is $1.81\times10^{-4}$ and $1.78\times10^{-4}$, respectively. At the original decimal-return scale, the contemporaneous term therefore dominates the joint estimation matrix. The heatmaps show structural differences associated with diagonal correction, but do not demonstrate a substantial additional benefit from the lagged term in this panel.

Figure~\ref{fig:sp500-main-sensitivity} illustrates sensitivity to the lag choice. For the full sample with $k_0=1$, the rule selects one factor, but the ratios at dimensions one and two are nearly tied ($0.2391$ and $0.2413$). With $k_0=2$, the first three ratios are $0.3712$, $0.5750$, and $0.3390$, selecting dimension three. Thus, the sector interpretations remain conditional on the sample and lag specification.

\section{Discussion}
\label{sec:Discussion}
This paper integrates contemporaneous covariance and lagged autocovariance information into loading-space estimation by applying HeteroPCA to a joint estimation matrix. This construction retains serial dependence information while correcting heterogeneous diagonal noise contamination. Under the stated assumptions, we establish consistency and convergence-rate bounds for dependent observations, identifying signal-strength regimes in which incorporating contemporaneous information improves the upper bound relative to the lagged-only estimator. The analysis also accommodates structured non-diagonal noise with known support through Generalized HeteroPCA. Simulations and an application to stock returns illustrate the estimation performance and loading-space structure of the proposed approach.

Several directions merit further investigation. Broader comparisons with alternative heteroskedastic factor estimators and low-rank recovery methods would help clarify the relative strengths of the proposed approach across noise and dependence structures. Such comparisons could assess both loading-space estimation and out-of-sample prediction. Other directions include data-adaptive weighting of contemporaneous and lagged information, joint selection of the factor number and lag order, and extensions to matrix- and tensor-valued time series.

\section*{Supplementary Material}
The Supplementary Material contains the detailed Experiment II results in Table~\ref{tab1}, full-size figures for the S\&P 500 stock-return application (Section~\ref{sec:supp-realdata-figures}), and proofs of Theorems~\ref{Hetep1} and~\ref{Hetep2}, the proof of Corollary~\ref{Hetep3}, and referenced theorems and lemmas (Section~\ref{sec:appendix}).

\begingroup
\bibhang=1.7pc
\bibsep=2pt
\fontsize{9}{14pt plus.8pt minus .6pt}\selectfont
\bibliographystyle{ss-chicago}
\nocite{bathia2010identifying,lam2011estimation}
\bibliography{manuscript-references}
\endgroup
\par\Needspace{11\baselineskip}
\vskip .65cm
{\fontsize{9}{14pt plus.8pt minus .6pt}\selectfont
\noindent Da Zhao (corresponding author)\\
Department of Biostatistics\\
Yale School of Public Health\\
New Haven, CT, USA\\
E-mail: \href{mailto:zhao.da@yale.edu}{zhao.da@yale.edu}

\vskip .3cm
\noindent Jing Chen\\
Department of Statistics and Data Science, National University of Singapore.\\
E-mail: \href{mailto:e0729952@u.nus.edu}{e0729952@u.nus.edu}\par}

\clearpage
\setcounter{section}{0}
\setcounter{subsection}{0}
\setcounter{equation}{0}
\setcounter{figure}{0}
\setcounter{table}{0}
\setcounter{lemma}{0}
\setcounter{quotedlemma}{0}
\setcounter{quotedtheorem}{0}
\renewcommand{\thefigure}{S\arabic{figure}}
\renewcommand{\thetable}{S\arabic{table}}
\renewcommand{\theHfigure}{S\arabic{figure}}
\renewcommand{\theHtable}{S\arabic{table}}
\renewcommand{\thesection}{S\arabic{section}}
\renewcommand{\thesubsection}{\thesection.\arabic{subsection}}
\renewcommand{\theequation}{\thesection.\arabic{equation}}
\renewcommand{\theHsection}{S\arabic{section}}
\renewcommand{\theHlemma}{supp.\arabic{lemma}}
\renewcommand{\theHquotedlemma}{supp.\arabic{quotedlemma}}
\renewcommand{\theHquotedtheorem}{supp.\arabic{quotedtheorem}}
\begin{center}\bfseries Supplementary Material\end{center}

{\fontsize{9}{11.5pt plus.8pt minus.6pt}\selectfont
\noindent This supplement contains the detailed simulation results in Table~\ref{tab1}, full-size figures for the S\&P 500 stock-return application (Section~\ref{sec:supp-realdata-figures}), and proofs of Theorems~\ref{Hetep1} and~\ref{Hetep2}, the proof of Corollary~\ref{Hetep3}, and referenced theorems and lemmas (Section~\ref{sec:appendix}).\par}
\fontsize{12}{14pt plus.8pt minus .6pt}\selectfont
\section{Experiment II Results}
\label{sec:supp-simulation-figures}

\begingroup
    \small
    \setlength{\tabcolsep}{5pt}
    \renewcommand{\arraystretch}{1.0}
    \begin{longtable}{rrrrrrr}
    \caption{Experiment II: mean projection errors over 200 replications per configuration. The smallest mean in each row is bold.}
    \label{tab1}\\
    \toprule
    $k_0$ & $n$ & $p$ & $\text{Error}_M$ & $\text{Error}_P$ & $\text{Error}_H$ & $\text{Error}_0$ \\
    \midrule
    \endfirsthead

    \multicolumn{7}{c}{Table~\thetable\ (continued)}\\
    \toprule
    $k_0$ & $n$ & $p$ & $\text{Error}_M$ & $\text{Error}_P$ & $\text{Error}_H$ & $\text{Error}_0$ \\
    \midrule
    \endhead
    \midrule
    \endfoot
    \bottomrule
    \endlastfoot
1 & 200 & 40 & 2.0283 & 2.0482 & \textbf{1.9386} & 2.0784 \\
1 & 200 & 100 & 1.8332 & 1.8381 & \textbf{1.7373} & 1.8485 \\
1 & 200 & 160 & 1.7263 & 1.7283 & \textbf{1.6389} & 1.7519 \\
1 & 200 & 240 & 1.6069 & 1.6076 & \textbf{1.5356} & 1.6382 \\
1 & 400 & 80 & 1.4757 & 1.4981 & \textbf{1.3355} & 1.4927 \\
1 & 400 & 200 & 1.1345 & 1.1382 & \textbf{1.0598} & 1.1350 \\
1 & 400 & 320 & 1.0386 & 1.0396 & \textbf{0.9883} & 1.0301 \\
1 & 400 & 480 & 0.9762 & 0.9763 & \textbf{0.9466} & 0.9528 \\
\midrule
2 & 200 & 40 & 1.9839 & 1.9984 & \textbf{1.8566} & 2.0784 \\
2 & 200 & 100 & 1.7635 & 1.7680 & \textbf{1.6391} & 1.8485 \\
2 & 200 & 160 & 1.6584 & 1.6607 & \textbf{1.5553} & 1.7519 \\
2 & 200 & 240 & 1.5518 & 1.5528 & \textbf{1.4655} & 1.6382 \\
2 & 400 & 80 & 1.3676 & 1.3847 & \textbf{1.2022} & 1.4927 \\
2 & 400 & 200 & 1.0689 & 1.0721 & \textbf{0.9907} & 1.1350 \\
2 & 400 & 320 & 1.0015 & 1.0025 & \textbf{0.9520} & 1.0301 \\
2 & 400 & 480 & 0.9572 & 0.9574 & \textbf{0.9264} & 0.9528 \\
\midrule
3 & 200 & 40 & 1.9637 & 1.9747 & \textbf{1.8185} & 2.0784 \\
3 & 200 & 100 & 1.7383 & 1.7423 & \textbf{1.5966} & 1.8485 \\
3 & 200 & 160 & 1.6408 & 1.6427 & \textbf{1.5227} & 1.7519 \\
3 & 200 & 240 & 1.5390 & 1.5400 & \textbf{1.4408} & 1.6382 \\
3 & 400 & 80 & 1.3325 & 1.3469 & \textbf{1.1545} & 1.4927 \\
3 & 400 & 200 & 1.0577 & 1.0602 & \textbf{0.9709} & 1.1350 \\
3 & 400 & 320 & 1.0005 & 1.0013 & \textbf{0.9480} & 1.0301 \\
3 & 400 & 480 & 0.9605 & 0.9606 & \textbf{0.9269} & 0.9528 \\
    \end{longtable}
\endgroup

\section{S\&P 500 Stock-Return Figures}
\label{sec:supp-realdata-figures}
\begin{figure}[!ht]
    \centering
    \includegraphics[page=1,width=0.75\linewidth]{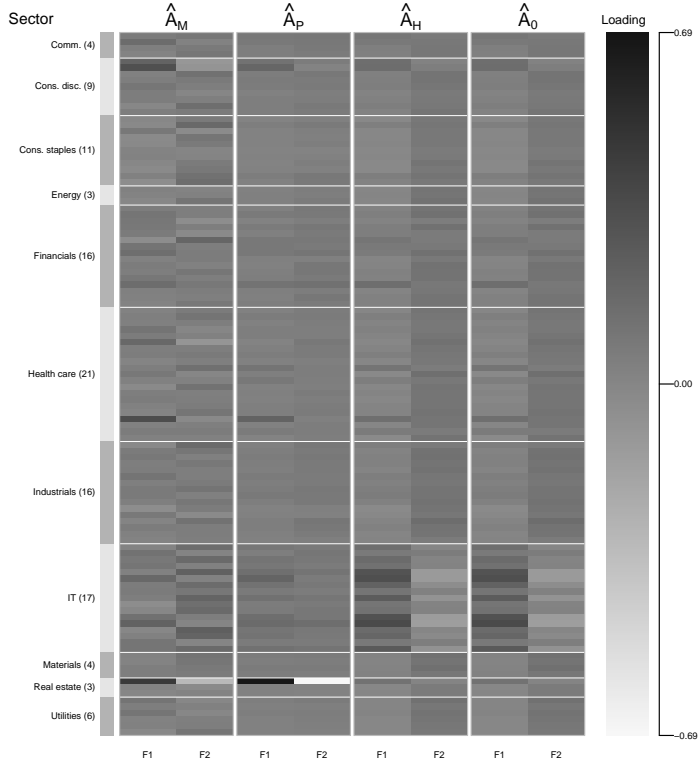}
    \caption{Varimax-rotated stock loading matrices for the initial training sample ($n=1144$, $p=110$, $k_0=1$, $\hat r=2$). The four panels correspond to the lagged-only, direct joint, joint HeteroPCA, and contemporaneous-only HeteroPCA estimators. Rows represent stocks grouped by current sector; numbers in parentheses give sector sizes. Comm., Cons. disc., Cons. staples, and IT denote Communication, Consumer Discretionary, Consumer Staples, and Information Technology, respectively. Lighter and darker shades indicate lower and higher signed loadings, respectively; the midpoint represents zero.}
    \label{fig:sp500-training-loadings}
\end{figure}
\clearpage

\begin{figure}[!ht]
    \centering
    \includegraphics[page=2,width=0.87\linewidth]{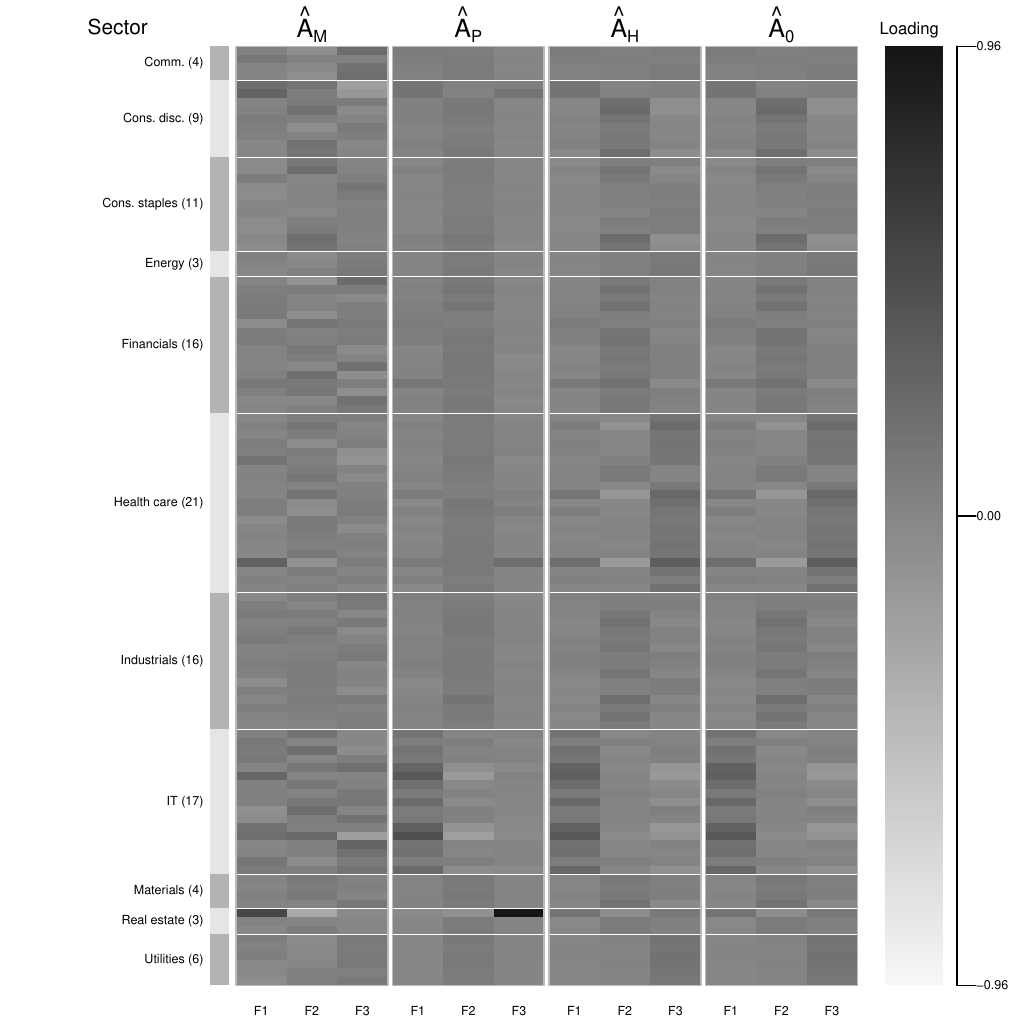}
    \caption{Varimax-rotated stock loading matrices for the full sample with two positive lags ($n=1634$, $p=110$, $k_0=2$, $\hat r=3$). Sector ordering and estimator notation are as in Figure~\ref{fig:sp500-training-loadings}. The common grayscale applies within this figure; its range differs from that in the other figures. These estimates use the entire panel and are intended for structural interpretation, not out-of-sample evaluation.}
    \label{fig:sp500-full-loadings}
\end{figure}
\clearpage

\begin{figure}[!ht]
    \centering
    \includegraphics[page=3,width=\linewidth]{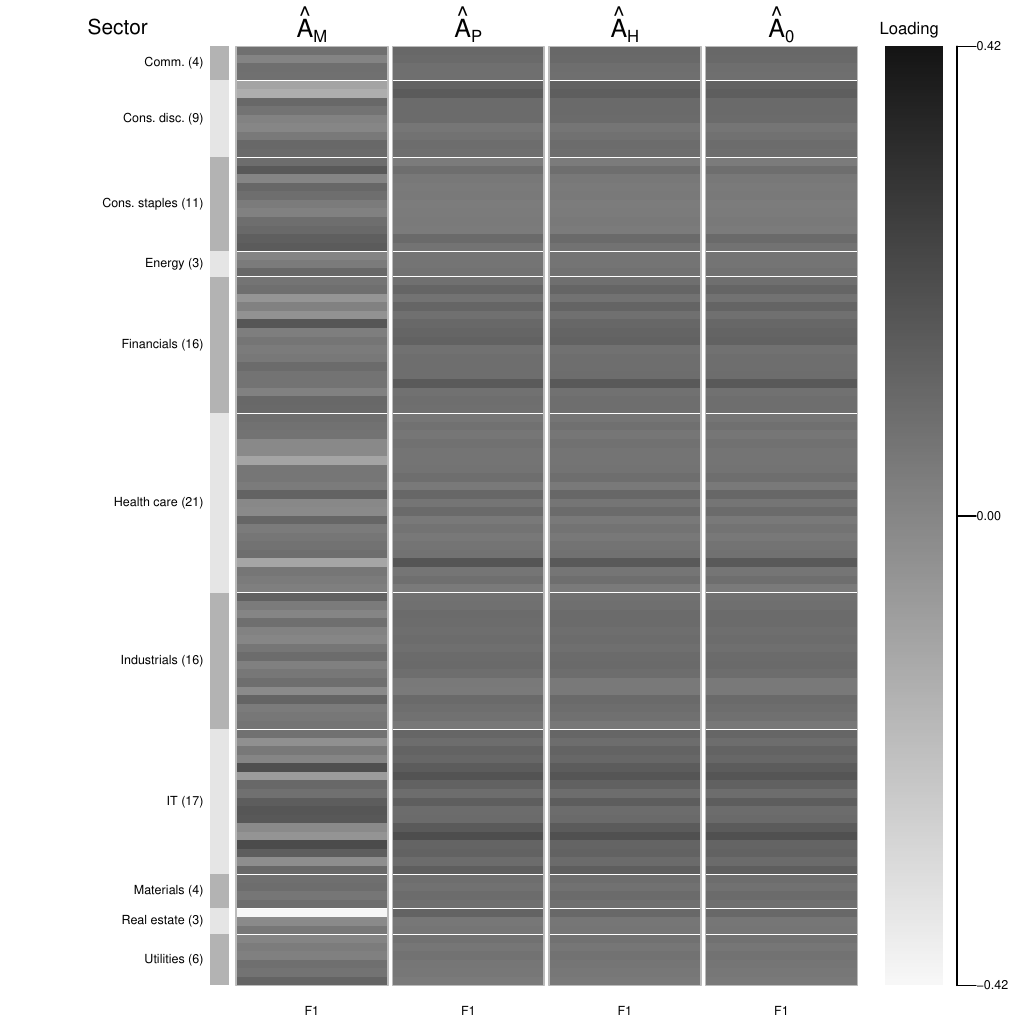}
    \caption{Full-sample sensitivity analysis with one positive lag ($n=1634$, $p=110$, $k_0=1$, $\hat r=1$). Sector ordering and estimator notation are as in Figure~\ref{fig:sp500-training-loadings}. The single loading direction requires no nontrivial varimax rotation. The selected dimension differs from the two-lag analysis; the grayscale is shared across methods within this figure only.}
    \label{fig:sp500-lag-sensitivity}
\end{figure}
\clearpage

\section{Proofs and Referenced Results}
\label{sec:appendix}

\subsection{Proof of Theorem \ref{Hetep1}}

\begin{proof}
First consider Algorithm \ref{algorithm2} run with the true rank $r$. The signal matrix $\bm M_0^*$ has rank $r$ and column space $\mathcal M(\bm A)$. If $\bm U$ is an orthonormal eigenbasis of this space, then $\bm U=\bm A\bm V$ for an orthogonal matrix $\bm V\in\bm O_{r,r}$, and $I(\bm U)=I(\bm A)$. Thus the incoherence condition in the cited result can be written using $\bm A$. On the sample perturbation event in Condition \ref{C6}, the proof of Theorem \ref{quoted:zhang-theorem7} in \cite{zhang2022heteroskedastic} gives, for $i=1,\ldots,T$,
\begin{equation}
\label{eq:proof-t1-iteration}
    \|\bm{N}^{(i)}-\bm{M}_0^*\|\leq 2\|\Gamma(\bm{N}^{(0)})-\Gamma(\bm{M}_0^*)\|+\|\bm{M}_0^*\|_{\min}(1/2)^{i+4}/\eta
\end{equation}

For comparison, write $\tau=\|\Gamma(\bm N^{(0)})-\Gamma(\bm M_0^*)\|$. If $\tau>0$ and $i$ is large enough that $\|\bm M_0^*\|_{\min}2^{-(i+4)}/\eta\leq2\tau$, the geometric remainder can be absorbed, giving
\begin{equation}
\label{eq:proof-t1-after-iterations}
\|\bm N^{(i)}-\bm M_0^*\|\leq4\tau
=4\|\Gamma(\bm N^{(0)})-\Gamma(\bm M_0^*)\|.
\end{equation}
This conditional bound is not needed when $\tau=0$; for the theorem's iteration rule we retain the geometric remainder in \eqref{eq:proof-t1-iteration} below.
\noindent We use the assumption that the noise covariance is supported on the known corruption set in the exact sense
\begin{equation*}
    \operatorname{supp}(\bm{\Sigma}_\epsilon)\subseteq\mathcal G,
    \qquad \Gamma(\bm{\Sigma}_\epsilon)=0.
\end{equation*}
Since $\bm{N}^{(0)}=\Gamma(\widehat{\bm{M}}^*)$, $\Gamma^2=\Gamma$, and $\bm{M}^*=\bm{M}_0^*+\bm{\Sigma}_\epsilon$, it follows that
\begin{equation}
\label{eq:proof-t1-projection-cancellation}
\begin{aligned}
    \Gamma(\bm{N}^{(0)})-\Gamma(\bm{M}_0^*)
    &=\Gamma(\widehat{\bm{M}}^*-\bm{M}_0^*)\\
    &=\Gamma(\widehat{\bm{M}}^*-\bm{M}^*)+\Gamma(\bm{\Sigma}_\epsilon)\\
    &=\Gamma(\widehat{\bm{M}}^*-\bm{M}^*).
\end{aligned}
\end{equation}
Under the same positive-$\tau$ remainder-absorption condition, applying Lemma 4 in \cite{zhang2022heteroskedastic} to $H=\widehat{\bm{M}}^*-\bm{M}^*$, with its rank rather than the number of factors $r$, and using $\sqrt{b\wedge\operatorname{rank}(H)}\leq\sqrt b$, we obtain
\begin{equation}
\label{eq:proof-t1-projected-error}
\begin{aligned}
    \|\bm{N}^{(i)}-\bm{M}_0^*\|
    &\leq 4\|\Gamma(\widehat{\bm{M}}^*-\bm{M}^*)\|\\
    &\leq 4(1+\sqrt b)\|\widehat{\bm{M}}^*-\bm{M}^*\|.
\end{aligned}
\end{equation}
Since $p$ is fixed and $b\leq p$, the factor $1+\sqrt b$ is bounded independently of $n$.
By the definitions of $\widehat{\bm{M}}^*$ and $\bm{M}^*$,
\begin{equation}
\label{eq:proof-t1-joint-sampling-decomposition}
    \|\widehat{\bm{M}}^*-\bm{M}^*\|\leq \sum_{k=1}^{k_0}\|\widehat{\bm{\Sigma}}_y(k)\widehat{\bm{\Sigma}}_y(k)^\top-{\bm{\Sigma}}_y(k){\bm{\Sigma}}_y(k)^\top\|+\|\widehat{\bm{\Sigma}}_y(0)-{\bm{\Sigma}}_y(0)\|
\end{equation}
We establish the sampling bounds directly from Condition \ref{C4}, with $p$ and $k_0$ fixed. Write $\bm\mu=E\bm y_t$, $\bm z_t=\bm y_t-\bm\mu$, and $q_{t,k}^{ab}=z_{t+k,a}z_{t,b}$ for $a,b\in\{1,\ldots,p\}$ and $0\leq k\leq k_0$. The fourth-moment assumption implies $E|q_{t,k}^{ab}|^2<\infty$. For $h>k$, the two products are measurable with respect to blocks separated by $h-k$ observations, so the $\psi$-mixing covariance inequality gives $|\operatorname{Cov}(q_{0,k}^{ab},q_{h,k}^{ab})|\leq C_{abk}\psi(h-k)$. The finitely many overlapping lags are controlled by Cauchy--Schwarz. Consequently,
\begin{equation*}
\begin{aligned}
\sum_{h\in\mathbb Z}
\left|\operatorname{Cov}(q_{0,k}^{ab},q_{h,k}^{ab})\right|&<\infty,\\
\operatorname{Var}\left\{\frac{1}{n-k}\sum_{t=1}^{n-k}q_{t,k}^{ab}\right\}
&=O(n^{-1}).
\end{aligned}
\end{equation*}
By Chebyshev's inequality and fixed $p$, the covariance estimator centered at the population mean satisfies
\begin{equation}
\label{eq:proof-t1-population-centered-covariance}
\widetilde{\bm\Sigma}_y(k)
=\frac{1}{n-k}\sum_{t=1}^{n-k}\bm z_{t+k}\bm z_t^\top,
\qquad
\|\widetilde{\bm\Sigma}_y(k)-\bm\Sigma_y(k)\|=O_p(n^{-1/2}).
\end{equation}
Define the full and truncated sample means by
\begin{equation}
\label{eq:proof-t1-centered-sample-means}
\begin{aligned}
\overline{\bm z}&=\frac{1}{n}\sum_{t=1}^n\bm z_t,\\
\overline{\bm z}_{+,k}&=\frac{1}{n-k}\sum_{t=1}^{n-k}\bm z_{t+k},
\qquad
\overline{\bm z}_{-,k}=\frac{1}{n-k}\sum_{t=1}^{n-k}\bm z_t.
\end{aligned}
\end{equation}
The same covariance-summability argument applied to $\bm z_t$ shows that all three means are $O_p(n^{-1/2})$. Thus replacing the population mean by the sample mean contributes only
\begin{equation}
\label{eq:proof-t1-sample-centering-error}
\begin{aligned}
\widehat{\bm\Sigma}_y(k)-\widetilde{\bm\Sigma}_y(k)
&=-\overline{\bm z}_{+,k}\overline{\bm z}^{\top}
  -\overline{\bm z}\,\overline{\bm z}_{-,k}^{\top}
  +\overline{\bm z}\,\overline{\bm z}^{\top},\\
\|\widehat{\bm\Sigma}_y(k)-\widetilde{\bm\Sigma}_y(k)\|&=O_p(n^{-1}).
\end{aligned}
\end{equation}
Define $\bm E_k=\widehat{\bm\Sigma}_y(k)-\bm\Sigma_y(k)$. Combining \eqref{eq:proof-t1-population-centered-covariance} and \eqref{eq:proof-t1-sample-centering-error}, and using fixed $k_0$, yields
\begin{equation}
\label{eq:proof-t1-uniform-covariance-rate}
\max_{0\leq k\leq k_0}\|\bm E_k\|=O_p(n^{-1/2}).
\end{equation}
For each $1\leq k\leq k_0$, expansion of the covariance product gives
\begin{equation}
\label{eq:proof-t1-lag-product-expansion}
\begin{aligned}
&\widehat{\bm\Sigma}_y(k)\widehat{\bm\Sigma}_y(k)^\top
-\bm\Sigma_y(k)\bm\Sigma_y(k)^\top\\
&\qquad=\bm E_k\bm\Sigma_y(k)^\top
+\bm\Sigma_y(k)\bm E_k^\top+\bm E_k\bm E_k^\top.
\end{aligned}
\end{equation}
Since the population covariance matrices have bounded norms in this fixed-dimensional setting,
\begin{equation*}
\begin{aligned}
&\|\widehat{\bm\Sigma}_y(k)\widehat{\bm\Sigma}_y(k)^\top
-\bm\Sigma_y(k)\bm\Sigma_y(k)^\top\|\\
&\qquad\leq 2\|\bm\Sigma_y(k)\|\,\|\bm E_k\|+\|\bm E_k\|^2
=O_p(n^{-1/2}).
\end{aligned}
\end{equation*}
Summing over the fixed number of lags gives
\begin{equation}
\label{eq:proof-t1-lag-product-rate}
\sum_{k=1}^{k_0}\|\widehat{\bm{\Sigma}}_y(k)\widehat{\bm{\Sigma}}_y(k)^\top-{\bm{\Sigma}}_y(k)\bm{\Sigma}_y(k)^\top\|=O_p(n^{-1/2}).
\end{equation}
Taking $k=0$ in \eqref{eq:proof-t1-uniform-covariance-rate} also gives
\begin{equation}
\label{eq:proof-t1-lag-zero-rate}
    \|\widehat{\bm{\Sigma}}_y(0)-{\bm{\Sigma}}_y(0)\|=O_p(n^{-1/2}).
\end{equation}
Combining \eqref{eq:proof-t1-joint-sampling-decomposition}, \eqref{eq:proof-t1-lag-product-rate}, and \eqref{eq:proof-t1-lag-zero-rate} gives $\|\widehat{\bm M}^*-\bm M^*\|=O_p(n^{-1/2})$. Lemma \ref{quoted:zhang-lemma4} then gives $\|\Gamma(\widehat{\bm M}^*-\bm M^*)\|=O_p(n^{-1/2})$, since $p$ and $b$ are fixed. Using the full iteration bound \eqref{eq:proof-t1-iteration}, the cancellation identity \eqref{eq:proof-t1-projection-cancellation}, and the iteration condition \eqref{eq:theorem1-iterations}, we obtain
\begin{equation}
\label{eq:proof-t1-imputed-matrix-rate}
\begin{aligned}
\|\bm N^{(T)}-\bm M_0^*\|
&\leq2\|\Gamma(\widehat{\bm M}^*-\bm M^*)\|
+\frac{\|\bm M_0^*\|_{\min}2^{-(T+4)}}{\eta}\\
&=O_p(n^{-1/2}).
\end{aligned}
\end{equation}
Here $\|\bm M_0^*\|_{\min}$ is a fixed population quantity. The argument remains valid when the projected perturbation is zero.

Next, we apply Lemma \ref{quoted:lam2011-lemma3} to obtain the loading-space error. Under Conditions \ref{C2} and \ref{C3}, $\mathcal M(\bm M)\subseteq\mathcal M(\bm A)$. Condition \ref{C5} gives $\operatorname{rank}(\bm M)=r$, so these spaces coincide. Define
\begin{equation*}
\begin{aligned}
D&=\bm A^\top\bm M_0^*\bm A
  =\bm A^\top\bm M\bm A+\bm\Sigma_x(0),\\
\bm M_0^*&=\bm A D\bm A^\top,\qquad
[\bm A\ \bm A^c]^\top\bm M_0^*[\bm A\ \bm A^c]
=\begin{pmatrix}D&0\\0&0\end{pmatrix}.
\end{aligned}
\end{equation*}
Since $\bm\Sigma_x(0)$ is positive semidefinite, $D\succeq\bm A^\top\bm M\bm A$. The nonzero eigenvalues of $\bm M$ are those of $\bm A^\top\bm M\bm A$. Thus, in the fixed population model,
\begin{equation}
\label{eq:proof-t1-signal-gap}
g:=\operatorname{sep}(D,0)=\lambda_{\min}(D)
=\|\bm M_0^*\|_{\min}\geq\lambda_r(\bm M)=:c_0>0.
\end{equation}
Here $D$ is an $r\times r$ signal matrix, distinct from the diagonal extraction operator $D(\cdot)$. It need not be diagonal; only the separation between its spectrum and zero is needed.

The exact support assumption and the sampling bounds above also show why the sample perturbation requirement in Condition \ref{C6} is satisfied in this fixed-dimensional setting:
\begin{equation*}
\frac{\eta\|\Gamma(\widehat{\bm M}^*)-\Gamma(\bm M_0^*)\|}
{\|\bm M_0^*\|_r}
=\frac{\eta\|\Gamma(\widehat{\bm M}^*-\bm M^*)\|}{g}
=O_p(n^{-1/2})=o_p(1).
\end{equation*}
Thus the event on which the iteration bounds were applied has probability tending to one.

Write $\bm F_n=\bm N^{(T)}-\bm M_0^*$. By \eqref{eq:proof-t1-imputed-matrix-rate} and \eqref{eq:proof-t1-signal-gap},
\begin{equation*}
\frac{\|\bm F_n\|}{g}=O_p(n^{-1/2})=o_p(1),
\qquad
\Pr\{\|\bm F_n\|\leq g/5\}\longrightarrow1.
\end{equation*}
The symmetry of $\mathcal G$ ensures that the iterates are symmetric. Apply Lemma \ref{quoted:lam2011-lemma3} with its population matrix equal to $\bm M_0^*$, perturbation equal to $\bm F_n$, and orthogonal blocks $Q_1=\bm A$, $Q_2=\bm A^c$. On the preceding event it yields a matrix $\bm P\in\mathbb R^{(p-r)\times r}$ such that
\begin{equation*}
\|\bm P\|
\leq\frac4g\|(\bm A^c)^\top\bm F_n\bm A\|
\leq\frac4g\|\bm F_n\|
=O_p(n^{-1/2}),
\end{equation*}
and the columns of
\begin{equation*}
\widetilde{\bm A}
=(\bm A+\bm A^c\bm P)(\bm I_r+\bm P^\top\bm P)^{-1/2}
\end{equation*}
form an orthonormal basis of an invariant subspace of $\bm N^{(T)}$.
By the spectral separation and Weyl's inequality, with probability tending to one this is the leading $r$-dimensional eigenspace. Hence, when the algorithm uses rank $r$, there exists $\bm R_n\in\bm O_{r,r}$ such that $\widetilde{\bm A}=\widehat{\bm A}_{GH}\bm R_n$.

Let $\bm C_n=(\bm I_r+\bm P^\top\bm P)^{-1/2}$. Then $\|\bm C_n\|\leq1$ and $\|\bm C_n-\bm I_r\|\leq\|\bm P\|^2/2$. Since $\|\bm P\|=o_p(1)$, with probability tending to one,
\begin{equation}
\label{eq:proof-t1-loading-rate}
\begin{aligned}
\min_{\bm R\in\bm O_{r,r}}
\|\widehat{\bm A}_{GH}\bm R-\bm A\|
&\leq\|\widetilde{\bm A}-\bm A\|\\
&\leq\|\bm C_n-\bm I_r\|+\|\bm P\|\,\|\bm C_n\|\\
&\leq\tfrac12\|\bm P\|^2+\|\bm P\|
\leq2\|\bm P\|=O_p(n^{-1/2}).
\end{aligned}
\end{equation}
This is a bound for an orthogonally aligned basis, not necessarily for the raw eigenvector matrix. It also implies the rotation-invariant bound
\begin{equation*}
\begin{aligned}
\|\widehat{\bm A}_{GH}\widehat{\bm A}_{GH}^\top-\bm A\bm A^\top\|
&=\|\widetilde{\bm A}\widetilde{\bm A}^\top-\bm A\bm A^\top\|\\
&\leq2\|\widetilde{\bm A}-\bm A\|=O_p(n^{-1/2}).
\end{aligned}
\end{equation*}
Finally, if $\Pr(\widehat r=r)\to1$, the estimated-rank algorithm agrees with its true-rank counterpart on this event (using the same iteration count and eigenvector-selection rule). For every $\varepsilon>0$, the probability of an incorrect rank is eventually less than $\varepsilon$, so the true-rank $O_p(n^{-1/2})$ projector bound transfers to the estimated-rank output. This proves the stated extension without comparing matrices with different numbers of columns.
\end{proof}

\subsection{Proof of Theorem \ref{Hetep2}}
\begin{proof}
The deterministic argument leading to \eqref{eq:proof-t1-projected-error} does not require fixed $p$. Under the stated rank and support assumptions, on the event in Condition \ref{C6} and for iteration counts satisfying the remainder-absorption condition used in \eqref{eq:proof-t1-after-iterations}, the same argument yields
\begin{equation}
\label{eq:proof-t2-projected-error}
\|\bm N^{(i)}-\bm M_0^*\|
\leq4\|\Gamma(\widehat{\bm M}^*-\bm M^*)\|.
\end{equation}
If the projected perturbation is zero, retain the geometric remainder in \eqref{eq:proof-t1-iteration} instead.

For comparison, Lemma \ref{quoted:zhang-lemma4}, applied with the rank of $\widehat{\bm M}^*-\bm M^*$ rather than the factor number $r$, gives the weaker bound, under the same remainder-absorption condition,
\begin{equation*}
\|\bm N^{(i)}-\bm M_0^*\|
\leq4(1+\sqrt b)\|\widehat{\bm M}^*-\bm M^*\|.
\end{equation*}
Unlike the fixed-dimensional case, $b$ may increase with $p$, so $1+\sqrt b$ cannot be absorbed into a constant without an additional assumption. We retain \eqref{eq:proof-t2-projected-error} as the primary bound for the subsequent high-dimensional analysis.

Under Condition \ref{C9}, Lemma \ref{Mdiff} bounds the sampling error around $\bm M^*$, not around $\bm M_0^*$. Since $\Gamma$ retains only entries outside $\mathcal G$, for every matrix $\bm H$,
\begin{equation*}
\|\Gamma(\bm H)\|\leq\|\Gamma(\bm H)\|_F\leq\|\bm H\|_F.
\end{equation*}
To use the iteration condition in Theorem \ref{Hetep2}, retain the geometric remainder in \eqref{eq:proof-t1-iteration}. Write $g_p=\|\bm M_0^*\|_{\min}$. Lemma \ref{Mmin} gives $g_p\asymp q_p$, so $2^{-T}/\eta=O(\rho_{n,p})$ implies that this remainder is $O(q_p\rho_{n,p})=O(p^{2-\delta}n^{-1/2})$. On the event in Condition \ref{C6}, Lemma \ref{Mdiff} consequently gives
\begin{equation*}
\begin{aligned}
\|\bm N^{(T)}-\bm M_0^*\|
&\leq2\|\Gamma(\widehat{\bm M}^*-\bm M^*)\|
+\frac{g_p2^{-(T+4)}}{\eta}\\
&\leq2\|\widehat{\bm M}^*-\bm M^*\|_F
+\frac{g_p2^{-(T+4)}}{\eta}\\
&=O_p\left(p^{2-\delta}n^{-1/2}+p^2n^{-1}\right).
\end{aligned}
\end{equation*}
The event in Condition \ref{C6} has probability tending to one. Since $p^\delta n^{-1/2}=O(1)$, the quadratic term is absorbed into the first term, yielding
\begin{equation}
\label{eq:proof-t2-imputed-simplified-rate}
\|\bm N^{(T)}-\bm M_0^*\|=O_p(p^{2-\delta}n^{-1/2}).
\end{equation}
No boundedness assumption on $b$ is used in this norm comparison, and the argument also applies when the projected perturbation is zero. The structural requirements of Condition \ref{C6} remain in force.

Next, define $D=\bm A^\top\bm M_0^*\bm A$ as in Lemma \ref{Mmin}. That lemma gives $\bm M_0^*=\bm A D\bm A^\top$ with $D$ positive definite. Therefore,
\begin{equation}
\label{eq:proof-t2-signal-gap}
g_p=\operatorname{sep}(D,0)=\lambda_{\min}(D)
=\|\bm M_0^*\|_{\min}
\geq c\,q_p
=c\max\{p^{2-2\delta},p^{1-\delta_0}\}.
\end{equation}
Let $\bm F_{n,p}=\bm N^{(T)}-\bm M_0^*$. Combining \eqref{eq:proof-t2-imputed-simplified-rate} with \eqref{eq:proof-t2-signal-gap} and the growth assumption $\rho_{n,p}\to0$ yields
\begin{equation*}
\frac{\|\bm F_{n,p}\|}{g_p}=O_p(\rho_{n,p})=o_p(1),
\qquad
\Pr\{\|\bm F_{n,p}\|\leq g_p/5\}\longrightarrow1.
\end{equation*}
The symmetry of $\mathcal G$ ensures that the iterates are symmetric. On this event, apply Lemma \ref{quoted:lam2011-lemma3} with population matrix $\bm M_0^*$, perturbation $\bm F_{n,p}$, and orthogonal blocks $\bm A,\bm A^c$. Exactly as in Theorem \ref{Hetep1}, there is a matrix $\bm P\in\mathbb R^{(p-r)\times r}$ satisfying
\begin{equation*}
\|\bm P\|
\leq\frac4{g_p}\|(\bm A^c)^\top\bm F_{n,p}\bm A\|
\leq\frac4{g_p}\|\bm F_{n,p}\|
=O_p(\rho_{n,p}),
\end{equation*}
and the corresponding invariant subspace has an orthonormal basis
\begin{equation*}
\widetilde{\bm A}
=(\bm A+\bm A^c\bm P)(\bm I_r+\bm P^\top\bm P)^{-1/2}.
\end{equation*}
By Weyl's inequality and separation from the zero eigenvalues, this is the leading $r$-dimensional eigenspace of $\bm N^{(T)}$. Hence $\widetilde{\bm A}=\widehat{\bm A}_{GH}\bm R_{n,p}$ for some $\bm R_{n,p}\in\bm O_{r,r}$. Using the same normalization bound as in \eqref{eq:proof-t1-loading-rate}, with probability tending to one,
\begin{equation*}
\begin{aligned}
\min_{\bm R\in\bm O_{r,r}}\|\widehat{\bm A}_{GH}\bm R-\bm A\|
&\leq\|\widetilde{\bm A}-\bm A\|\\
&\leq\|\bm P\|+\tfrac12\|\bm P\|^2
\leq2\|\bm P\|=O_p(\rho_{n,p}).
\end{aligned}
\end{equation*}
The rotation-invariant error obeys
\begin{equation*}
\|\widehat{\bm A}_{GH}\widehat{\bm A}_{GH}^\top-\bm A\bm A^\top\|
\leq2\|\widetilde{\bm A}-\bm A\|=O_p(\rho_{n,p}).
\end{equation*}
Finally, the deterministic rate simplifies exactly to
\begin{equation*}
\begin{aligned}
\rho_{n,p}
&=\frac{p^{2-\delta}n^{-1/2}}
{\max\{p^{2-2\delta},p^{1-\delta_0}\}}\\
&=\min\{p^\delta n^{-1/2},p^{1+\delta_0-\delta}n^{-1/2}\}.
\end{aligned}
\end{equation*}
This proves \eqref{eq:theorem2-aligned-and-projector-rate}. The second expression is strictly smaller in its power of $p$ when $2\delta>1+\delta_0$; otherwise the first expression determines the bound. Without the quadratic-term absorption condition, the sampling calculation retains the numerator $p^{2-\delta}n^{-1/2}+p^2/n$, so the displayed two-rate conclusion cannot be obtained by simply dropping $p^2/n$.
\end{proof}

\begin{lemma}
\label{Mmin}
Under Conditions \ref{C2}, \ref{C3}, \ref{C7}, and \ref{C8}, with fixed $k_0\geq1$, there are constants $c,C>0$, independent of $p$, such that for all sufficiently large $p$,
\begin{equation}
\label{eq:lemma1-signal-eigenvalue-bounds}
c\max\{p^{2-2\delta},p^{1-\delta_0}\}
\leq\|\bm M_0^*\|_{\min}
\leq\|\bm M_0^*\|
\leq C\max\{p^{2-2\delta},p^{1-\delta_0}\}.
\end{equation}
\end{lemma}
\begin{proof}
For $k=1,\ldots,k_0$, define the $r\times p$ matrices
\begin{equation*}
\bm B_k=\bm\Sigma_x(k)\bm A^\top+\bm\Sigma_{x\epsilon}(k),
\qquad
\bm\Sigma_y(k)=\bm A\bm B_k.
\end{equation*}
Under Conditions \ref{C2} and \ref{C3}, the contemporaneous factor--noise cross-covariance is zero. Consequently,
\begin{equation}
\label{eq:lemma1-signal-decomposition}
\begin{aligned}
D&=\bm\Sigma_x(0)+\sum_{k=1}^{k_0}\bm B_k\bm B_k^\top,\\
\bm M_0^*&=\bm A D\bm A^\top,\qquad
D=\bm A^\top\bm M_0^*\bm A.
\end{aligned}
\end{equation}
Equivalently, horizontal concatenation gives
\begin{equation*}
\bm L=[\bm\Sigma_x^{1/2}(0)\bm A^\top,\bm B_1,\ldots,\bm B_{k_0}],
\qquad D=\bm L\bm L^\top.
\end{equation*}
This is not the product of sums of the blocks.

Since $\bm A^\top\bm A=\bm I_r$, the $r$ singular values of $\bm\Sigma_x(k)\bm A^\top$ equal those of $\bm\Sigma_x(k)$. By the singular-value perturbation inequality and Conditions \ref{C7}--\ref{C8}, there are constants $c_1,C_1>0$ such that, for all sufficiently large $p$ and each $k=1,\ldots,k_0$,
\begin{equation*}
\begin{aligned}
\sigma_r(\bm B_k)
&\geq\sigma_r\{\bm\Sigma_x(k)\}
-\|\bm\Sigma_{x\epsilon}(k)\|
\geq c_1p^{1-\delta},\\
\|\bm B_k\|
&\leq\|\bm\Sigma_x(k)\|+\|\bm\Sigma_{x\epsilon}(k)\|
\leq C_1p^{1-\delta}.
\end{aligned}
\end{equation*}
The lower bound is positive before it is squared. Each summand in \eqref{eq:lemma1-signal-decomposition} is positive semidefinite, so
\begin{equation*}
\begin{aligned}
\lambda_{\min}(D)
&\geq\lambda_{\min}\{\bm\Sigma_x(0)\}
+\sum_{k=1}^{k_0}\sigma_r(\bm B_k)^2\\
&\geq c_0p^{1-\delta_0}+k_0c_1^2p^{2-2\delta}\\
&\geq c\max\{p^{2-2\delta},p^{1-\delta_0}\}.
\end{aligned}
\end{equation*}
Similarly, using fixed $k_0$,
\begin{equation*}
\begin{aligned}
\|D\|
&\leq\|\bm\Sigma_x(0)\|+\sum_{k=1}^{k_0}\|\bm B_k\|^2\\
&\leq C\max\{p^{2-2\delta},p^{1-\delta_0}\}.
\end{aligned}
\end{equation*}
Thus $D$ is positive definite, $\bm M_0^*$ has rank $r$, and its nonzero eigenvalues are exactly those of $D$. This proves \eqref{eq:lemma1-signal-eigenvalue-bounds}.
\end{proof}

\begin{lemma}
\label{Mdiff}
Under Conditions \ref{C2}, \ref{C3}, \ref{C4}, and \ref{C7}--\ref{C9}, let $r$ and $k_0$ be fixed as $n,p\to\infty$. For clarity, Condition \ref{C9} states that, with $\bm z_t=\bm y_t-E\bm y_t$ and mixing coefficient $\psi_p(h)$ at dimension $p$,
\begin{equation}
\label{eq:lemma2-uniform-moment-mixing}
\sup_p\max_{1\leq a\leq p}E|z_{t,a}|^4\leq C,
\qquad
\sup_p\sum_{h=1}^{\infty}h\,\psi_p(h)^{1/2}\leq C.
\end{equation}
Then
\begin{equation*}
\|\widehat{\bm M}^*-\bm M^*\|_F
=O_p\left(p^{2-\delta}n^{-1/2}+p^2n^{-1}\right).
\end{equation*}
In particular, if $p^\delta n^{-1/2}=O(1)$, then
\begin{equation}
\label{eq:lemma2-joint-simplified-rate}
\|\widehat{\bm M}^*-\bm M^*\|
\leq\|\widehat{\bm M}^*-\bm M^*\|_F
=O_p(p^{2-\delta}n^{-1/2}).
\end{equation}
\end{lemma}
\begin{proof}
The sampling error must be centered at $\bm M^*$, since
\begin{equation*}
\widehat{\bm M}^*-\bm M_0^*
=(\widehat{\bm M}^*-\bm M^*)+\bm\Sigma_\epsilon.
\end{equation*}
Put $\bm S_k=\bm\Sigma_y(k)$ and $\bm E_k=\widehat{\bm\Sigma}_y(k)-\bm\Sigma_y(k)$. The same product expansion as in \eqref{eq:proof-t1-lag-product-expansion} gives
\begin{equation*}
\widehat{\bm M}^*-\bm M^*
=\bm E_0+\sum_{k=1}^{k_0}
\left(\bm E_k\bm S_k^\top+\bm S_k\bm E_k^\top
+\bm E_k\bm E_k^\top\right).
\end{equation*}
Using $\|\bm U\bm V\|_F\leq\|\bm U\|_F\|\bm V\|$ and $\|\bm E_k\|\leq\|\bm E_k\|_F$, we obtain
\begin{equation}
\label{eq:lemma2-joint-error-bound}
\|\widehat{\bm M}^*-\bm M^*\|_F
\leq\|\bm E_0\|_F+
\sum_{k=1}^{k_0}
\left(2\|\bm S_k\|\,\|\bm E_k\|_F+\|\bm E_k\|_F^2\right).
\end{equation}

Firstly, for $k=1,\ldots,k_0$, Conditions \ref{C2}, \ref{C3}, \ref{C7}, and \ref{C8} imply
\begin{equation}
\label{eq:lemma2-population-lag-bound}
\begin{aligned}
\|\bm S_k\|
&=\|\bm A\bm\Sigma_x(k)\bm A^\top
+\bm A\bm\Sigma_{x\epsilon}(k)\|\\
&\leq\|\bm\Sigma_x(k)\|+\|\bm\Sigma_{x\epsilon}(k)\|
=O(p^{1-\delta}).
\end{aligned}
\end{equation}

Secondly, we repeat the sampling argument of Theorem \ref{Hetep1}, now keeping the dimension dependence explicit. Let $q_{t,k}^{ab}=z_{t+k,a}z_{t,b}$. By \eqref{eq:lemma2-uniform-moment-mixing} and Cauchy--Schwarz, $E|q_{t,k}^{ab}|^2$ is bounded uniformly in $p,a,b$ and $0\leq k\leq k_0$. For $h>k$, the $\psi$-mixing covariance inequality bounds the covariance of two such products by $C_1\psi_p(h-k)$, where $C_1$ does not depend on $p,a,b,k$. The finitely many overlapping lags have a uniform bound as well. Consequently, for constants independent of the dimension and coordinates,
\begin{equation*}
\begin{aligned}
\sup_p\max_{\substack{1\leq a,b\leq p\\0\leq k\leq k_0}}
\sum_{h\in\mathbb Z}
|\operatorname{Cov}(q_{0,k}^{ab},q_{h,k}^{ab})|
&\leq C_2,\\
\operatorname{Var}\left\{\frac{1}{n-k}
\sum_{t=1}^{n-k}q_{t,k}^{ab}\right\}
&\leq \frac{C_3}{n}.
\end{aligned}
\end{equation*}
Here and below these bounds hold for all sufficiently large $n$. For the population-centered estimator $\widetilde{\bm\Sigma}_y(k)$ defined in \eqref{eq:proof-t1-population-centered-covariance}, summing the entrywise variances and applying Markov's inequality gives
\begin{equation*}
\begin{aligned}
E\|\widetilde{\bm\Sigma}_y(k)-\bm\Sigma_y(k)\|_F^2
&\leq C_3p^2/n,\\
\max_{0\leq k\leq k_0}
\|\widetilde{\bm\Sigma}_y(k)-\bm\Sigma_y(k)\|_F
&=O_p(pn^{-1/2}).
\end{aligned}
\end{equation*}
The maximum is over a fixed number of lags; it is not an entrywise maximum.

For the sample means in \eqref{eq:proof-t1-centered-sample-means}, the same uniform covariance-summability argument, applied to $z_{t,a}$, yields
\begin{equation*}
\|\overline{\bm z}\|
+\max_{0\leq k\leq k_0}
\left(\|\overline{\bm z}_{+,k}\|+\|\overline{\bm z}_{-,k}\|\right)
=O_p\left(\sqrt{p/n}\right).
\end{equation*}
Using the exact centering identity in \eqref{eq:proof-t1-sample-centering-error} and the Frobenius norm of a rank-one matrix, we have
\begin{equation*}
\begin{aligned}
\|\widehat{\bm\Sigma}_y(k)-\widetilde{\bm\Sigma}_y(k)\|_F
&\leq\|\overline{\bm z}_{+,k}\|\,\|\overline{\bm z}\|
+\|\overline{\bm z}\|\,\|\overline{\bm z}_{-,k}\|
+\|\overline{\bm z}\|^2\\
&=O_p(p/n).
\end{aligned}
\end{equation*}
Therefore,
\begin{equation}
\label{eq:lemma2-uniform-covariance-frobenius-rate}
\max_{0\leq k\leq k_0}\|\bm E_k\|_F=O_p(pn^{-1/2}).
\end{equation}
This treats lag zero as a centered sampling error, without dropping $\bm\Sigma_\epsilon$ or imposing separate mixing conditions on the latent factors and noise.

Finally, substituting \eqref{eq:lemma2-population-lag-bound} and \eqref{eq:lemma2-uniform-covariance-frobenius-rate} into \eqref{eq:lemma2-joint-error-bound}, and using fixed $k_0$ and $0\leq\delta\leq1$, gives
\begin{equation*}
\begin{aligned}
\|\widehat{\bm M}^*-\bm M^*\|_F
&=O_p\left(pn^{-1/2}+p^{2-\delta}n^{-1/2}+p^2n^{-1}\right)\\
&=O_p\left(p^{2-\delta}n^{-1/2}+p^2n^{-1}\right).
\end{aligned}
\end{equation*}
Since $(p^2n^{-1})/(p^{2-\delta}n^{-1/2})=p^\delta n^{-1/2}$, the additional growth restriction gives \eqref{eq:lemma2-joint-simplified-rate}.
\end{proof}

\subsection{Proof of Corollary \ref{Hetep3}}
\begin{proof}
Condition \ref{C1} implies Condition \ref{C3} and ensures that $\bm\Sigma_\epsilon$ is supported on the diagonal set $\mathcal G=\{(i,i):1\leq i\leq p\}$. This set is symmetric and satisfies $b=\eta=1$ by \eqref{eq:diagonal-mask-parameters}. Condition \ref{C10} is therefore the population incoherence part of Condition \ref{C6}. By Lemmas \ref{Mmin} and \ref{Mdiff}, and the growth assumptions,
\begin{equation*}
\begin{aligned}
\frac{\eta\|\Gamma(\widehat{\bm M}^*)-\Gamma(\bm M_0^*)\|}
{\|\bm M_0^*\|_r}
&=\frac{\|\Delta(\widehat{\bm M}^*-\bm M^*)\|}
{\|\bm M_0^*\|_r}\\
&\leq\frac{2\|\widehat{\bm M}^*-\bm M^*\|}
{\|\bm M_0^*\|_r}
=O_p(\rho_{n,p})=o_p(1).
\end{aligned}
\end{equation*}
Thus the sample perturbation part of Condition \ref{C6} holds with probability tending to one. All assumptions of Theorem \ref{Hetep2} are satisfied, and the two algorithms coincide for this corruption set. Applying that theorem gives \eqref{eq:corollary-diagonal-subspace-rate}.
\end{proof}

\subsection{Referenced Theorems and Lemmas}
\label{sec:referenced-results}

\setcounter{quotedtheorem}{0}
\begin{quotedtheorem}[Bathia, Yao and Ziegelmann (2010), Theorem 1]
\label{quoted:bathia2010-theorem1}
Let conditions C1--C4 hold. Then as $n\to\infty$, the following assertions hold:
\begin{enumerate}
\renewcommand{\labelenumi}{(\roman{enumi})}
\item $\|\widehat K-K\|_{\mathcal S}=O_p(n^{-1/2})$.
\item For $j=1,\ldots,d$, $|\widehat\theta_j-\theta_j|=O_p(n^{-1/2})$ and
\begin{equation*}
\left(\int_{\mathcal I}\{\widehat\psi_j(u)-\psi_j(u)\}^2\,du\right)^{1/2}
=O_P(n^{-1/2}).
\end{equation*}
\item For $j\geq d+1$, $\widehat\theta_j=O_p(n^{-1})$.
\item Let $\{\psi_j:j\geq d+1\}$ be a complete orthonormal basis of $\mathcal M^\perp$, and put
\begin{equation*}
f_j(\cdot)=\sum_{i=d+1}^{\infty}
\langle\psi_i,\widehat\psi_j\rangle\psi_i(\cdot).
\end{equation*}
Then for any $j\geq d+1$,
\begin{equation*}
\begin{aligned}
&\left(\int_{\mathcal I}
\left\{\sum_{i=1}^{d}\langle\psi_i,\widehat\psi_j\rangle\psi_i(u)\right\}^2\,du
\right)^{1/2}\\
&\qquad=\left(\int_{\mathcal I}\{\widehat\psi_j(u)-f_j(u)\}^2\,du\right)^{1/2}
=O_p(n^{-1/2}).
\end{aligned}
\end{equation*}
\end{enumerate}
\end{quotedtheorem}

\setcounter{quotedlemma}{0}
\begin{quotedlemma}[Lam, Yao and Bathia (2011), Lemma 1]
\label{quoted:lam2011-lemma1}
Under (3) with Conditions 1--7 in \S\S~2.1 to 2.3, we have
\begin{equation*}
\begin{aligned}
\|\Sigma_f(k)\|_2&\asymp N^{1-\delta}\asymp\|\Sigma_f(k)\|_{\min},\\
\|\Sigma_{f,\epsilon}(k)\|_2&=O(N^{1-\delta/2}).
\end{aligned}
\end{equation*}
\end{quotedlemma}

\begin{quotedlemma}[Lam, Yao and Bathia (2011), Lemma 2]
\label{quoted:lam2011-lemma2}
Under (3) and Condition 5 in \S~2.1, we have for $0\leq k\leq k_0$,
\begin{equation*}
\|\widetilde\Sigma_f(k)-\Sigma_f(k)\|_2=O_P(N^{1-\delta}T^{-1/2}),
\end{equation*}
\begin{equation*}
\|\widetilde\Sigma_\epsilon(k)-\Sigma_\epsilon(k)\|_2=O_P(NT^{-1/2}),
\end{equation*}
\begin{equation*}
\begin{aligned}
\|\widetilde\Sigma_{f,\epsilon}(k)-\Sigma_{f,\epsilon}(k)\|_2
&=O_P(N^{1-\delta/2}T^{-1/2})\\
&=\|\widetilde\Sigma_{\epsilon,f}(k)-\Sigma_{\epsilon,f}(k)\|_2.
\end{aligned}
\end{equation*}
Moreover, $\|f_t\|_2^2=O_P(N^{1-\delta})$ for all integers $t\geq0$.
\end{quotedlemma}

\begin{quotedlemma}[Lam, Yao and Bathia (2011), Lemma 3, p.~915]
\label{quoted:lam2011-lemma3}
Suppose $A$ and $A+E$ are $n\times n$ symmetric matrices and that
$Q=[Q_1\ Q_2]$, where $Q_1$ has size $n\times r$ and $Q_2$ has size
$n\times(n-r)$, is an orthogonal matrix such that $\operatorname{span}(Q_1)$
is an invariant subspace for $A$; that is,
$A\times\operatorname{span}(Q_1)\subset\operatorname{span}(A)$.
Partition the matrices $Q^\top A Q$ and $Q^\top E Q$ as follows:
\begin{equation*}
 Q^\top A Q=
 \begin{pmatrix}D_1&0\\0&D_2\end{pmatrix},
 \qquad
 Q^\top E Q=
 \begin{pmatrix}E_{11}&E_{21}^\top\\E_{21}&E_{22}\end{pmatrix}.
\end{equation*}
If
\begin{equation*}
 \operatorname{sep}(D_1,D_2)
 =\min_{\lambda\in\lambda(D_1),\,\mu\in\lambda(D_2)}
       |\lambda-\mu|>0,
\end{equation*}
where $\lambda(M)$ denotes the set of eigenvalues of the matrix $M$,
and $\|E\|_2\leq\operatorname{sep}(D_1,D_2)/5$, then there exists a
matrix $P\in\mathbb R^{(n-r)\times r}$ with
\begin{equation*}
 \|P\|_2\leq\frac{4}{\operatorname{sep}(D_1,D_2)}\|E_{21}\|_2
\end{equation*}
such that the columns of
$\widehat Q_1=(Q_1+Q_2P)(I+P^\top P)^{-1/2}$ define an orthonormal
basis for a subspace that is invariant for $A+E$.
\end{quotedlemma}

\setcounter{quotedtheorem}{6}
\begin{quotedtheorem}[Zhang, Cai and Wu, Theorem 7: General robust $\sin\Theta$ theorem]
\label{quoted:zhang-theorem7}
Assume $\mathcal G\in[p]\times[p]$ is $b$-sparse. Suppose one observes
the symmetric matrix $N=M+Z$, where $\operatorname{rank}(M)=r$,
$Z$ is any symmetric perturbation, and the eigenvectors of $M$ are
$U\in\mathbb O_{p,r}$. Let $\widehat U^{(t)}$ be the intermediate
matrix in Algorithm 1 of the cited paper with $t$ iterations. There exists a constant
$c>0$ such that if the incoherence condition
\begin{equation*}
 \frac{I(U)\|M\|}{\lambda_r(M)}
 \leq\frac{cp}{\eta br(b\wedge r)}
\end{equation*}
is satisfied and $\eta\|\Gamma(Z)\|\leq c\lambda_r(M)$, then
\begin{equation*}
 \|\sin\Theta(\widehat U^{(t)},U)\|
 \leq4\|\Gamma(Z)\|/\lambda_r(M)+2^{-(t+3)}/\eta.
\end{equation*}
Here, $\eta$ is defined in equation (34) of \citet{zhang2022heteroskedastic}, corresponding to~\eqref{eq:mask-amplification} in the main paper. In particular, if
\begin{equation*}
 T=\Omega\left(\log\frac{\lambda_r(M)}
                    {\eta\|\Gamma(Z)\|}\vee1\right),
\end{equation*}
the final outcome $\widehat U$ of Algorithm 1 of the cited paper with corrupted index
set $\mathcal G$ satisfies
\begin{equation*}
 \|\sin\Theta(\widehat U,U)\|
 \lesssim \frac{\|\Gamma(Z)\|}{\lambda_r(M)}\wedge1.
\end{equation*}
\end{quotedtheorem}

\begin{quotedlemma}[Zhang, Cai and Wu, supplementary Lemma 4]
\label{quoted:zhang-lemma4}
If $M\in\mathbb R^{m_1\times m_2}$, $\operatorname{rank}(M)=r$,
$\mathcal G\subseteq[m_1]\times[m_2]$,
$\max_i|\{j:(i,j)\in\mathcal G\}|\leq b$,
$\max_j|\{i:(i,j)\in\mathcal G\}|\leq b$, then we have
\begin{equation*}
 \|G(M)\|\leq\sqrt{b\wedge r}\|M\|,
 \qquad
 \|\Gamma(M)\|\leq(\sqrt{b\wedge r}+1)\|M\|.
\end{equation*}
In particular, if $M\in\mathbb R^{p\times p}$ is any square matrix
and $\Delta(M)$ is the matrix $M$ with diagonal entries set to 0, then
\begin{equation*}
 \|\Delta(M)\|\leq2\|M\|.
\end{equation*}
Here, the factor \lq\lq 2\rq\rq{} in the statement above cannot be improved.
\end{quotedlemma}
\end{document}